\documentclass[letterpaper, 12pt]{article}

\usepackage{rotating} 

\usepackage{mystyle}

\DeclareMathAlphabet\mathbb{U}{fplmbb}{m}{n} %Preferred style of \mathbb

\newcommand{\re}{\mathbb{R}}

\DeclareMathOperator{\Exp}{\mathbb{E}}
\DeclareMathOperator{\Var}{Var}
\DeclareMathOperator{\Prob}{\mathbb{P}}

\DeclareMathOperator{\interior}{int}

\DeclareMathOperator{\supp}{supp}

\DeclareMathOperator{\boundary}{bd}

\newcommand{\plim}{\mathchoice{\stackrel{\mathbb{P}}{\rightarrow}}{\rightarrow_\mathbb{P}}{\rightarrow_\mathbb{P}}{\rightarrow_\mathbb{P}}}

\newcommand{\absSmall}[1]{\vert #1 \vert}
\newcommand{\norm}[1]{\left\Vert #1 \right\Vert}
\newcommand{\normSmall}[1]{\Vert #1 \Vert}

\newcommand{\Indic}[1]{\mathbb{1}\left[#1\right]}
\newcommand{\IndicSmall}[1]{\mathbb{1}[#1]}

\newcommand{\op}[1]{o_{\Prob}(#1)}
\newcommand{\Op}[1]{O_{\Prob}(#1)}

\newcommand{\tth}[1]{#1^{\text{th}}}

\newcommand\independenT{\protect\mathpalette{\protect\independeNT}{\perp}}
\def\independeNT#1#2{\mathrel{\rlap{$#1#2$}\mkern2mu{#1#2}}}
\DeclareMathOperator{\independent}{\independenT}

\newcommand{\R}{\ensuremath{\mathcal{R}}}

\usepackage[longnamesfirst]{natbib} 

\usepackage{fullpage}

\def\MYTITLE{Instrumental Variables Estimation of a Generalized Correlated Random 
    Coefficients Model}

\def\MYKEYWORDS{correlated random coefficients, instrumental variables, 
    unobserved heterogeneity, semiparametrics, hedonic models, residential sorting, valuation of clean air}

\hypersetup{
	pdftitle = \MYTITLE,
	pdfauthor = Matthew A. Masten and Alexander Torgovitsky,
	pdfkeywords = \MYKEYWORDS,
	pdfnewwindow = true
}

\title{\MYTITLE\footnote{This paper was presented at the UW--Milwaukee 
        Mini-Conference on Microeconometrics, Duke, Northwestern, Chicago, the 
        Harvard/MIT joint workshop, Ohio State, Boston College, the Triangle 
        econometrics conference, University College London, and the First Conference 
        in Econometric Theory at Universidad de San Andr\'{e}s. We thank 
        audiences at those seminars as well as Federico Bugni, Ivan Canay, 
        Matias Cattaneo, Andrew Chesher, Bryan Graham, Jim Heckman, Stefan Hoderlein, Joel 
        Horowitz, Max Kasy, Shakeeb Khan, Jia Li, Arnaud Maurel, Whitney Newey, 
        Maya Rossin-Slater, Elie Tamer, Duncan Thomas, Chris Timmins, and Ed 
        Vytlacil for helpful conversations and comments.  We thank Fu Ouyang for 
        research assistance.}}

\author{Matthew A. Masten\footnote{Department of Economics, Duke University, 
        \texttt{matt.masten@duke.edu}} \qquad Alexander 
    Torgovitsky\thanks{
    Department of Economics, Northwestern University, 
    \texttt{a-torgovitsky@northwestern.edu}}
}

\date{January 1, 2014}

\begin{document}
\maketitle
\begin{abstract}
We study identification and estimation of the average treatment effect in a correlated random coefficients model that allows for first stage heterogeneity and binary instruments. The model also allows for multiple endogenous variables and interactions between endogenous variables and covariates. Our identification approach is based on averaging the coefficients obtained from a collection of ordinary linear regressions that condition on different realizations of a control function. This identification strategy suggests a transparent and computationally straightforward estimator of a trimmed average treatment effect constructed as the average of kernel-weighted linear regressions. We develop this estimator and establish its $\sqrt{n}$--consistency and asymptotic normality. Monte Carlo simulations show excellent finite-sample performance that is comparable in precision to the standard two-stage least squares estimator. We apply our results to analyze the effect of air pollution on house prices, and find substantial heterogeneity in first stage instrument effects as well as heterogeneity in treatment effects that is consistent with household sorting.
\end{abstract}



%%%%%%%%%%%%%%%%%%%%%%%%%%%%%%%%%%%%%%%%%%%%%%%%%%%%%%%%%%%%%%%%%%%%%%%%%%%%%%%%
%%%%%%%%%%%%%%%%%%%%%%%%%%%%%%%%%%%%%%%%%%%%%%%%%%%%%%%%%%%%%%%%%%%%%%%%%%%%%%%%
%%%%%%%%%%%%%%%%%%%%%%%%%%%%%%%%%%%%%%%%%%%%%%%%%%%%%%%%%%%%%%%%%%%%%%%%%%%%%%%%
%%%%%%%%%%%%%%%%%%%%%%%%%%%%%%%%%%%%%%%%%%%%%%%%%%%%%%%%%%%%%%%%%%%%%%%%%%%%%%%%
%%%%%%%%%%%%%%%%%%%%%%%%%%%%%%%%%%%%%%%%%%%%%%%%%%%%%%%%%%%%%%%%%%%%%%%%%%%%%%%%

\onehalfspacing

\newpage
\normalsize
\section{Introduction}

This paper is about the linear correlated random coefficients (CRC) model.  In 
its simplest form, the model can be written as
\begin{equation}
    \label{eq_model_simple}
    Y = B_{0} + B_{1}X,
\end{equation}
where $Y$ is an outcome variable, $X$ is an explanatory variable and $B \equiv 
(B_0, B_1)$ are unobservable variables. The explanatory variable $X$ is 
endogenous in the sense that it may be statistically dependent with $B_{0}$ and 
$B_{1}$.  Concerns about endogeneity are widespread in economic applications and 
are often addressed by using the variation of an instrumental variable, $Z$, 
that is plausibly independent (or uncorrelated) with $(B_{0}, B_{1})$, but 
correlated with $X$. The most common tool for doing this is the two-stage least 
squares (TSLS) estimator.  However, unless the partial effect of $X$ on $Y$, 
i.e.  $B_{1}$, is a degenerate random variable (a constant), the estimand of the 
TSLS estimator is not necessarily an easily interpretable feature of the 
distribution of $B_{1}$.  Assuming that $B_{1}$ is constant is tantamount to 
assuming that the treatment effect of $X$ on $Y$ is homogenous. As many authors 
have discussed, the theoretical and empirical evidence does not support the 
assumption of homogenous treatment effects. See \cite{heckman2001tjope} and 
\cite{imbens2007ic} for thorough expositions of this point.

To address this problem, several authors have explored auxiliary assumptions 
under which the TSLS estimand becomes a parameter of interest. The most 
influential finding is that of \cite{imbensangrist1994e}, who show that if both 
$X$ and $Z$ are binary and if $Z$ affects $X$ monotonically, then the TSLS 
estimator is consistent for the local average treatment effect (LATE). The LATE 
parameter has generated significant debate over whether it is actually a 
quantity that economists should be interested in; see, for example, the 
\textit{Journal of Economic Perspectives} (2010) and the \textit{Journal of 
    Economic Literature} (2010) symposia. However, as the support of $X$ grows 
from binary to multi-valued discrete to continuous, the TSLS estimand becomes an 
increasingly complicated weighted average of LATEs between different $X$ 
realizations (\citealt{angristimbens1995jotasa}).  Even if one finds a solitary 
LATE to be an interesting parameter, the interpretation, economic significance, 
and policy relevance of such weighted averages of LATEs is more tenuous. A 
less controversial parameter of interest is the average treatment effect (ATE), 
which due to the linearity in \eqref{eq_model_simple} is determined by the 
average partial effect (APE), $\Exp(B_{1})$.  In a series of papers, 
\cite{heckmanvytlacil1998tjohr} and 
\cite{wooldridge1997el,wooldridge2003el,wooldridge2008ic} showed that if the 
effect of $Z$ on $X$ is homogenous, then TSLS will be consistent for 
$\Exp(B_{1})$.  This type of homogeneity assumption is somewhat unsatisfying, 
since accounting for heterogeneity is the main motivation for considering this 
problem to begin with. (See also \cite{litobias2011joe}, who consider Bayesian 
inference in those models.)

An alternative is to consider different instrumental variables estimators 
besides TSLS. \cite*{florensheckmanmeghiretal2008e} take this approach in 
considering a polynomial version of \eqref{eq_model_simple} plus an additive 
nonparametric function of $X$ common to all units. They show that the ATE is 
identified if $X$ is continuously distributed and there exists a function $h$ 
that is strictly increasing in a scalar unobservable $V$ such that $X = h(Z,V)$.  
This type of first stage restriction allows for heterogeneity in the effect of 
$Z$ on $X$, albeit in a limited form, and so directly addresses the concerns 
about previous work by Heckman, Vytlacil and Wooldridge. The utility of the 
first stage restriction is in creating a random variable $R$ which is a control 
function in the sense that $X \independent B \vert R$. A central contribution of 
our paper is to exploit this control function property to provide an alternate 
identification approach to the one considered by 
\cite{florensheckmanmeghiretal2008e}. Our approach has three main benefits 
relative to that of \cite{florensheckmanmeghiretal2008e}.  First, while 
\cite{florensheckmanmeghiretal2008e} require a continuous instrument (see the 
discussion on page \pageref{florens_cts_req_remark}), we can achieve 
identification with binary and discrete instruments in many cases.  Second, our 
approach enables us to include multiple endogenous variables, non-polynomial 
terms and interactions between endogenous variables and covariates in more 
general linear-in-coefficients specifications of \eqref{eq_model_simple}. Third, 
it suggests a computationally straightforward estimator that appears to have 
good finite sample properties. The main drawbacks of our approach relative to 
that of \cite{florensheckmanmeghiretal2008e} is that their model allows for a 
common additive nonparametric function of $X$, and, when $Z$ is continuous, 
their ``measurable separability'' assumption may hold in some cases that our 
corresponding relevance condition does not.

Our results build on recent research on nonparametric identification in 
nonseparable models. A recurring finding in this work is a trade-off between the 
dimension of heterogeneity and the required variation in the instrument $Z$. At 
one extreme lie the papers by \cite{imbensnewey2009e} and \cite{kasy2013wp}, who 
show that unrestricted forms of heterogeneity can be allowed in the outcome 
and/or first stage equations while still attaining point identification of the 
ATE, as long as $Z$ satisfies a large support assumption (they also provide sharp partial identification results when the large support assumption does not hold). Despite their ubiquity across the econometric 
theory literature, such large support assumptions are unlikely to ever be even 
approximately satisfied in practice.  In particular, they rule out the binary 
and discrete instruments that are commonly found in applied work, such as policy 
shifts, institutional changes, and natural experiments. On the other hand, work by 
\cite{chernozhukovhansen2005e}, \cite{torgovitsky2012wpa} and 
\cite{d'haultfouillefevrier2012wp} has shown that binary and discrete 
instruments of this sort can secure identification, as long as the dimension of 
heterogeneity is sufficiently restricted. These restrictions on heterogeneity 
rule out simple, parsimonious specifications like \eqref{eq_model_simple} which 
contain more than one unobservable. Between these two extremes lies the paper by 
\cite{florensheckmanmeghiretal2008e} and also those of \cite{chesher2003e} and 
\cite{masten2012jmp}, both of which require a continuous instrument with small 
support but also allow for additional heterogeneity. Our paper contributes to 
this middle ground and, among other things, provides an example where a broadly 
interesting parameter can be identified in a model with high-dimensional 
heterogeneity and discrete instruments.

The recent work of \cite{grahampowell2012e} (who build on work by 
\citealt{chamberlain1992e}) on CRC models with panel data is related in 
motivation to this paper. Both papers seek to identify the APE---at least among 
some subpopulation---but the analysis is fundamentally different due to 
differences between using panels and instruments as sources of identification.  
Partially related to their paper as well as ours is the literature on random 
uncorrelated coefficient models; for example, \cite{beranhall1992taos} and 
\cite{hoderleinklemelamammen2010et}. That literature assumes $X$ and $(B_0, 
B_1)$ are independent and centers on estimating the distribution of $(B_{0}, 
B_{1})$. In contrast, we limit our focus to identifying averages, but have to 
contend with the difficulty of dependence between $X$ and $(B_{0}, B_{1})$.

An advantage of our identification approach and the linear structure in 
\eqref{eq_model_simple} is that it facilitates estimators that are precise, easy 
to implement, and which do not suffer from the curse of dimensionality. A main 
contribution of our paper is to develop such an estimator of $\Exp(B)$ and 
establish its $\sqrt{n}$--consistency and asymptotic normality. (Due to 
uniformity issues, we actually develop asymptotic theory for an estimator of a 
trimmed version of $\Exp(B)$; see section \ref{sec_asymptotics}.) Our estimator 
is essentially an average of ordinary linear regressions run conditional on a 
realization of a control function and so shares similarities with the control 
function approaches of, for example, \cite{blundellpowell2004troes}, 
\cite{imbensnewey2009e}, \cite{rothe2009joe}, \cite{hoderleinsherman2013cwp4} 
and \cite{torgovitsky2013wp}. The control function is estimated with a first 
stage quantile or distribution regression and the conditioning is approximated 
with kernel weights. Hence, our estimator reduces to a straightforward average 
of weighted linear regressions, where the weights are determined by a first 
stage quantile or distribution regression of $X$ on $Z$. Incorporating 
covariates is a simple matter of including them in these linear mean and 
quantile regressions.  Monte Carlo experiments show that our estimator can 
perform as well or better than the TSLS estimator under conditions when both 
would be consistent, while remaining consistent in situations where TSLS would 
be inconsistent.

We apply our results to study the effect of air pollution on house prices. We 
follow the empirical approach of \cite{chaygreenstone2005jope}, who argue that 
instrumenting is necessary to deal with unobserved economic shocks and sorting 
of households based on unobserved preferences for clean air. They also argue 
that this sorting leads to correlated random coefficients. They define a binary 
instrument based on regulation implemented by the 1970 Clean Air Act Amendments.  
We demonstrate substantial first stage heterogeneity in the effect of this 
instrument, which strongly suggests that the simpler estimators discussed by 
\cite{heckmanvytlacil1998tjohr} and 
\cite{wooldridge1997el,wooldridge2003el,wooldridge2008ic} would be inconsistent for the APE. Likewise, the binary instrument precludes approaches which rely on continuous variation, such as \cite{florensheckmanmeghiretal2008e}.  
For two subsets of counties where the instrument has a statistically significant 
effect on pollution levels, we estimate unweighted APEs of changes in pollution 
on changes in house prices. These estimates demonstrate patterns that are 
consistent with household sorting. Taken together, these estimates along with 
TSLS suggest there is substantial heterogeneity in households' valuation of 
clean air.

The structure of the paper is as follows. In the next section we formally 
discuss the model, assumptions and our identification results. In Section 
\ref{sec_est}, we describe our estimator and discuss its implementation. In 
Section \ref{sec_asymptotics}, we analyze the asymptotic properties of our 
estimator. In Section \ref{sec_mc}, we report the results of Monte Carlo studies 
that demonstrate the performance of our estimator. Finally, in Section 
\ref{sec_app}, we present our application to air pollution and house prices.  
Section \ref{sec_conclusion} concludes.

\section{Model and Identification}
\label{sec_model}

A general version of model \eqref{eq_model_simple} is
\begin{equation}
    \label{eq_model}
    Y = B_{0} + \sum_{j=1}^{d_{x}} B_{j}X_{j} + \sum_{j=1}^{d_{1}} B_{d_{x} + 
        j}Z_{1j} \equiv W'B,
\end{equation}
where $X \in \re^{d_{x}}$ is a vector of potentially endogenous variables, 
$Z_{1} \in \re^{d_{1}}$ is a vector of included exogenous variables with 
$\tth{j}$ component $Z_{1j}$, $W \equiv [1, X', Z_{1}']' \in \re^{d_{w}}$ with 
$d_{w} \equiv 1 + d_{x} + d_{1}$, and $B \in \re^{d_{w}}$ is a vector of 
unobservable variables. In addition to $Z_{1}$, there is a vector of excluded exogenous 
variables (instruments) $Z_{2} \in \re^{d_{2}}$ that do not directly affect $Y$ 
in \eqref{eq_model}. We write the exogenous variables together as $Z \equiv 
[Z_{1}', Z_{2}']' \in \re^{d_{z}}$ with $d_{z} \equiv d_{1} + d_{2}$. 

We divide the vector of endogenous variables into subvectors of lengths $d_{b} 
\geq 1$ and $d_{x} - d_{b} \geq 0$. We refer to the first $d_{b}$ components of 
$X$ as the \emph{basic} endogenous variables and the last $d_{x} - d_{b}$ 
components of $X$ as the \emph{derived} endogenous variables. We assume that the 
basic endogenous variables satisfy a particular first stage structure that is 
specified in the assumptions ahead. In contrast, the derived endogenous 
variables are assumed to be functions of the basic endogenous variables and the 
included exogenous variables $Z_{1}$. For example, we could have $d_{b} = 1$ and 
derived endogenous variables $X_{k} = X^{k}$ for $k > d_{b}$, as in the model of 
\cite{florensheckmanmeghiretal2008e}. Or, we could have $X_{k} = X_{1}Z_{1}$ for 
some $k > d_{b}$ be an interaction variable, which would allow the distribution 
of partial effects to differ arbitrarily across values of $Z_{1}$. This allows, 
for example, men and women to have different distributions of treatment effects, 
allowing for heterogeneity on observables to be dealt with in the usual way.

Throughout our analysis we frequently use the random vector
\[
    R \equiv [F_{X_{1} \vert Z}(X_{1} \vert Z),\ldots,F_{X_{d_{b}} \vert 
        Z}(X_{d_{b}} \vert Z)]',
\]
which we refer to as the \emph{conditional rank} of $X$. We are only concerned 
with the conditional ranks of the basic endogenous variables, since under our 
assumptions the conditional ranks of the derived endogenous variables $F_{X_k 
    \mid Z}(X_k \vert Z)$ for $k=d_b+1,\ldots, d_x$ will contain less 
information. Below, we will restrict $X_{k}$ to be continuously distributed for 
$k = 1,\ldots,d_{b}$ so that $R_{k}$ is distributed uniformly on $[0,1]$ for 
these $k$. Note, however, that if $d_{b} > 1$ then the support of $R$ may be a 
proper subset of $[0,1]^{d_{b}}$. Consider the following assumptions.

\begin{enumerate}[align=left, leftmargin=*, widest=I4.]
    \itshape
    \renewcommand{\theenumi}{Assumption I}
    \renewcommand{\labelenumi}{\textbf{\theenumi.}}
    \item
    \setcounter{enumi}{0}

    \renewcommand{\theenumi}{I\arabic{enumi}}
    \renewcommand{\labelenumi}{\textbf{\theenumi.}}
    \item \label{as_moments} \textbf{(Existence of moments)} $\Exp( B ) < 
        \infty$ and $\Exp( \| W \|^2 ) < \infty$.

    \item \label{as_fs} \textbf{(First stage equation)} For each basic 
        endogenous variable $k = 1,\ldots,d_{b}$, there exists a scalar random 
        variable $V_k$ and a possibly unknown function $h_{k}$ that is strictly increasing in its second argument, for which $X_{k} = h_{k}(Z, V_{k})$. The vector $V \equiv (V_{1},\ldots,V_{d_{b}})$ is continuously distributed.

    \item \label{as_derived} \textbf{(Derived endogenous variables)} For each $k 
        = d_{b} + 1,\ldots,d_{x}$, there exists a known function $g_{k}$ such 
        that $X_{k} = g_{k}(X_{1}\,\ldots,X_{d_{b}}, Z_{1})$.

    \item \label{as_ex} \textbf{(Instrument exogeneity)} $(B, V) \independent 
        Z$.

    \item \label{as_rel} \textbf{(Instrument relevance)} $\Exp[WW' \vert R = r]$ 
        is invertible for almost every $r$ in a known Lebesgue measurable set 
        $\R \subseteq \supp(R)$.
\end{enumerate}

\begin{theorem}
    \label{thm_id}
    Define $\beta(r) \equiv \Exp[B \vert R = r]$.  Under Assumptions I,
    \begin{align*}
        \beta(r) = \Exp[WW' \vert R = r]^{-1}\Exp[WY \vert R = r]
    \end{align*}
    for any $r \equiv (r_{1},\ldots,r_{d_{b}}) \in \R$. Hence both $\beta(r)$ 
    and $\beta_{\R} \equiv \Exp[B \vert R \in \R]$ are point identified.
\end{theorem}

The proof of Theorem \ref{thm_id} uses the following implication of \ref{as_fs} 
and \ref{as_ex}, which has been used and analyzed in various forms by 
\cite{imbens2007ic}, \cite{florensheckmanmeghiretal2008e}, 
\cite{imbensnewey2009e}, \cite{kasy2010et} and \cite{torgovitsky2012wpa}. Since 
our version is a slight extension, we provide a short proof in the appendix.

\begin{proposition}
    \label{prop_controlvariable}
    \ref{as_fs} and \ref{as_ex} imply that $(R, B) 
    \independent Z$. If \ref{as_derived} also holds, then $W \independent B \vert R$.
\end{proposition}

\begin{proof}[\textbf{Proof of Theorem \ref{thm_id}}]
\ref{as_moments} ensures that all conditional moments of interest exist. Premultiplying both sides of \eqref{eq_model} by $W$ and taking expectations 
    conditional on $R = r$ for any $r \in \R$, we have
    \begin{align*}
        \Exp[WY \vert R = r] = \Exp[WW'B \vert R = r] = \Exp[WW' \vert R = 
        r]\beta(r),
    \end{align*}
    by Proposition \ref{prop_controlvariable}. Given \ref{as_rel}, we can 
    premultiply both sides by the inverse of $\Exp[WW' \vert R = r]$ to obtain 
    the claimed expression for $\beta(r)$. Since $\Exp[WW' \vert R = r]^{-1}$ 
    and $\Exp[WY \vert R = r]$ are features of the observable data, this shows 
    that $\beta(r)$ and $\beta_{\R} \equiv \Exp[B \vert R \in \R]$ are both 
    identified for any known $\mathcal{R} \subseteq \supp(R)$ satisfying 
    \ref{as_rel}.
\end{proof}

    The intuition behind Theorem \ref{thm_id} is as follows. After conditioning 
    on $R = r$, all of the variation in the basic endogenous variables is due to 
    variation in $Z$, by the definition of $R$. Since the derived endogenous 
    variables are functions of the basic endogenous variables and $Z_{1}$, all 
    of the variation in $W$ conditional on $R = r$ is also due to variation in 
    $Z$. Variation in $Z$, however, is independent of $B$ conditional on $R = r$ 
    by instrument exogeneity (\ref{as_ex}) via Proposition 
    \ref{prop_controlvariable}. As a result, a linear regression of $Y$ on $X$ 
    conditional on $R = r$ identifies $\beta(r) \equiv \Exp[B \vert R = r]$.  
    Averaging $\Exp[B \vert R = r]$ over $r \in \mathcal{R}$ then yields 
    $\beta_{\R} \equiv \Exp[B \vert R \in \mathcal{R}]$. If instrument relevance 
    (\ref{as_rel}) holds for some measure one subset of $\supp(R)$, then 
    $\beta_{\R} = \Exp[B]$ is identified. This intuition is illustrated in 
    Figure \ref{fig_id_intuition}.

    \begin{figure}[h]
        \centering
        \input{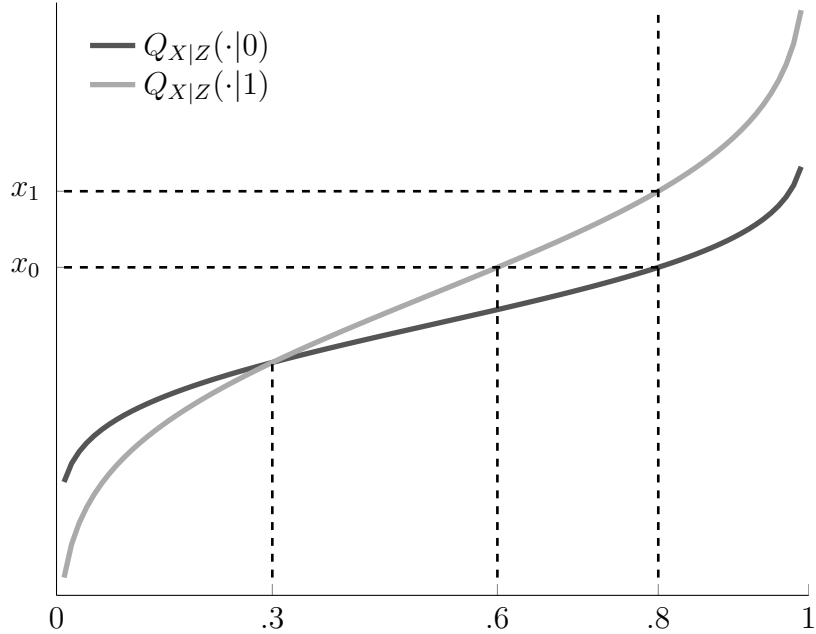}
        \caption{\label{fig_id_intuition} Consider the simple CRC model 
            $\eqref{eq_model_simple}$, $Y = B_0 + B_1 X$, and suppose $Z$ is 
            binary. Conditional on $R = 0.8$, $X$ assumes two values ($x_0$ and 
            $x_1$) depending on the realization of $Z$.  Since $Z \independent B 
            \mid \{ R = 0.8 \}$, a mean regression of $Y$ on $X$ conditional on $R = 
            0.8$ identifies the means of the intercept and slope 
            coefficients, $\Exp(B_0 \mid R=0.8)$ and $\Exp(B_1 \mid R=0.8)$. For 
            the plotted quantile functions, the relevance condition \ref{as_rel} 
            holds for almost every $r \in (0,1)$, since the curves intersect at 
            only one point. Hence the previous argument yields $\Exp(B \mid 
            R=r)$ for all $r \in (0,1)$ and averaging then gives $\Exp(B)$. Note 
            also that the instrument's effect is nonmonotonic---it is positive 
            for units with large $R$ (above $R=0.3$) and negative for units with 
            small $R$.}
    \end{figure}

Theorem \ref{thm_id} is complementary to a result by 
\cite{florensheckmanmeghiretal2008e}. Those authors consider a model with a 
single basic endogenous variable $X$ and the outcome equation
\[
    Y = \varphi(X) + B_0 + B_1 X + B_2 X^2 + \cdots + B_K X^K,
\]
for some pre-specified $K$, where $\varphi$ is an unknown function, and 
$(B_0,\ldots,B_K)$ are random coefficients that are potentially correlated with 
$X$. Except for $\varphi$, this outcome equation can be obtained from 
\eqref{eq_model} with basic endogenous variable $X$, and derived endogenous 
variables $(X^{2},\ldots,X^{K})$. The price of including the $\varphi$ function 
is that \cite{florensheckmanmeghiretal2008e} require a continuous small support 
instrument (see their identification proof on page 1203, the step from equation 
10 to the next line).\label{florens_cts_req_remark} We do not include the 
$\varphi$ function, but are generally able to achieve identification of the 
average coefficients in the polynomial outcome equation model so long as the 
distribution of $Z$ has at least $K + 1$ support points.  
\cite{florensheckmanmeghiretal2008e} also maintain \ref{as_fs} and \ref{as_ex}, 
but in place of \ref{as_rel} they impose a ``measurable separability'' condition 
that is somewhat high-level. As those authors discuss, their measurable 
separability condition may fail if the first stage equation is not continuous in 
$V$. In contrast, our relevance condition \ref{as_rel} does not require such 
continuity. This allows for the support of $X$ conditional on $Z$ to be disjoint.

\ref{as_rel} is directly analogous to the standard no-multicollinearity 
condition in ordinary least squares and consequently requires the analyst to 
avoid standard causes of failure, such as the dummy variable trap. When $d_x = 
d_b = 1$, so that there is a single basic endogenous variable and no derived 
endogenous variables, \ref{as_rel} requires that $\Var[Q_{X \mid Z}(r \mid Z)] > 
0$ for all $r \in \mathcal{R}$. If $Z \in \{0,1\}$ is binary, then $\Var[Q_{X 
    \mid Z}(r \mid Z)] > 0$ happens if and only if $Q_{X \mid Z}(r \mid 0) \neq 
Q_{X \mid Z}(r \mid 1)$; that is, the two curves in Figure 
\ref{fig_id_intuition} are separated at $r$. Since $Q_{X \mid Z}(r \mid Z) = 
h[Z,Q_V(r)]$ by strict monotonicity (\ref{as_fs}) and independence 
(\ref{as_ex}), we must have that for each $r \in \mathcal{R}$ there are distinct 
$z, z' \in \supp(Z)$ with $h[z,Q_V(r)] \neq h[z', Q_V(r)]$. Hence, for all units 
with first stage unobservables $v = Q_V(r)$ for which we want to learn 
$\Exp(B)$, the instrument must affect those units' endogenous variable. 
Generally, whether \ref{as_rel} holds is an empirical matter in the sense that 
the condition only depends on the distribution of observables and so, at least 
in principle, can be checked in the data. 

    When \ref{as_rel} only holds for some proper subset $\R$ of $\supp(R)$ then 
    Theorem \ref{thm_id} identifies $\beta_{\R} \equiv \Exp[B \vert R \in 
    \mathcal{R}]$, which generally will not equal $\Exp[B]$. Nevertheless, 
    $\beta_{\R}$ has an interpretation similar to the unweighted LATE of 
    \cite{imbensangrist1994e}. That is, $\beta_{\R}$ is the unweighted average 
    of $B$ for those agents for whom the instrument has an effect. Note that we 
    do not require this effect to be monotonic. If $Z$ is assumed to have a 
    monotonic effect on $X$ (as in \citealt{imbensangrist1994e}), then 
    $\beta_{\R}$ is the unweighted average of $B$ for those agents who are 
    induced to increase their treatment intensity $X$ due to a change in $Z$.  
    This type of parameter may be of comparable (or even greater) interest than 
    $\Exp[B]$ for a policy maker considering a policy change that affects the 
    determination of $X$ through an incentive $Z$.

    While \ref{as_rel} may fail for some subset of $\supp(R)$, it is an 
    intuitively appealing requirement for an instrument. Agents characterized by 
    an $r$ at which $\Exp[WW' \vert R = r]$ is singular do not experience 
    independent variation in $W$ due to variation in $Z$, and so it is natural 
    that $\Exp[B \vert R = r]$ should not be identifiable for those agents.  
    Assuming that the effect of $Z$ on $X$ is homogenous, as in 
    \cite{heckmanvytlacil1998tjohr} and 
    \cite{wooldridge1997el,wooldridge2003el,wooldridge2008ic}, ignores this 
    distinction and explicitly includes agents in the average for whom the 
    instrument might be completely ineffectual---in effect, extrapolating from 
    $Z$-sensitive agents to $Z$-insensitive agents. Similarly, the measurable 
    separability condition of \cite{florensheckmanmeghiretal2008e} could 
    apparently hold even if there is a non-negligible subset of agents for whom 
    the instrument is irrelevant.

    As in standard linear regression analysis, identification of $\Exp(B)$ (or 
    some conditional version of it) via Theorem \ref{thm_id} provides 
    identification of the ATE and APE when the outcome equation includes 
    nonlinear functions of $X$ or interactions with covariates $Z_{1}$.  This is 
    an elementary point, but we mention it for clarity. Suppose that $Y = B_{0} 
    + B_{1}X + B_{2}XZ_{1}$.  Then the APE is given by $\Exp[B_{1}] + 
    \Exp[B_{2}]\Exp[Z_{1}]$ while the ATE for an exogenous change from $x$ to 
    $\overline{x}$ is given by $(\Exp[B_{1}] + 
    \Exp[B_{2}]\Exp[Z_{1}])(\overline{x} - x)$. Both of these quantities can be 
    obtained from estimates of $\Exp[B_{1}]$, $\Exp[B_{2}]$ and
    $\Exp[Z_{1}]$. Alternatively, an analyst may be interested in the APE for 
    some predetermined value of $z_{1}$, which would be given by $\Exp[B_{1}] + 
    \Exp[B_{2}]z_{1}$. When \eqref{eq_model} contains nonlinear terms, e.g. $Y = 
    B_{0} + B_{1}X + B_{2}X^{2}$, then an analyst may be more interested in 
    reporting $\Exp[B_{1}] + 2\Exp[B_{2}]x$ as the APE when $X$ is exogenously 
    set to $x$.  All of these quantities can be obtained after applying our 
    identification results.

    Among the maintained assumptions for Theorem \ref{thm_id}, \ref{as_fs} is 
    generally the most controversial. While it is more flexible than the 
    homogenous effect specifications of \cite{heckmanvytlacil1998tjohr} and 
    \cite{wooldridge1997el,wooldridge2003el,wooldridge2008ic}, it does restrict 
    the basic endogenous variables to be continuous and also limits the 
    heterogeneity in their first stage equations to have dimension one.  
    One-dimensional heterogeneity of the sort in \ref{as_fs} can be interpreted 
    as ``rank invariance'' in the effect of $Z$ on each basic component of
    $X$. (The concept of rank invariance was first introduced by 
    \citealt{doksum1974taos}.) Rank invariance means that the ordinal ranking of 
    any two agents in terms of any component of $X_{k}$ ($k \leq d_{b}$) would 
    be the same if both agents received the same realization of $Z$, for any 
    realization of $Z$. See \cite{heckmansmithclements1997troes}, 
    \cite{chernozhukovhansen2005e} and \cite{torgovitsky2012wpa} for further 
    discussions of rank invariance. While one-dimensional heterogeneity is 
    restrictive, there are few alternatives in the literature that allow for 
    high-dimensional heterogeneity in both the outcome and first stage equations 
    while attaining point identification of a broadly interpretable parameter.  
    An important exception to this is the work of \cite{kasy2013wp}, who obtains 
    such a result but under the assumption that $Z$ affects a scalar $X$ 
    monotonically and also has large support.
    
Assumptions \ref{as_fs} and \ref{as_ex} together generally imply that 
correlation between $X$ and the random coefficients $B$ must occur through $V$. 
For example, specifying $B$ as a direct function of $X$, such as setting $B_0 = 
X$, implies that, conditional on $R$, some variation remaining in $B$ is due to 
$Z$, and hence \ref{as_ex} will typically not hold. Thus, \ref{as_fs} and 
\ref{as_ex} should be viewed as also placing restrictions on the manner in which 
$X$ and $B$ may be dependent. This point is not unique to our model---even in a 
simple textbook model \eqref{eq_model_simple} with a constant $B_1$, data 
generating processes like $B_0 = X$ will violate the usual uncorrelatedness 
assumption $\Exp(Z B_0) = 0$. Consequently, if we wish to express model 
\eqref{eq_model_simple} in terms of potential outcomes, it is helpful to view 
$(B_0,B_1,V)$ as unobserved heterogeneity parameters which are intrinsic to each 
unit. After the instrument is assigned, the value of $X$ is determined via the 
first stage equation of \ref{as_fs} and then the value of $Y$ is determined 
through \eqref{eq_model_simple}. Thus the average partial effect $\Exp(B_1)$ 
tells us the average effect of exogenously increasing $X$ by one for all units.

In addition to the overall average of $\Exp(B)$ identified in Theorem 
\ref{thm_id}, the following result shows that averages for groups determined by 
their treatment intensity are also identified. This parameter is analogous to 
the ``effect of the treatment on the treated'' parameter defined in 
\cite{florensheckmanmeghiretal2008e}. 

\begin{theorem}\label{thm_att}
Under Assumptions I, the ``average effect of treatment on the treated'' 
parameter $\Exp(B \mid X_{k} = x_{k}, k \leq d_{b})$ is point identified for any 
$x = (x_1,\ldots,x_{d_b}) \in \supp(X_{1},\ldots,X_{d_{b}})$ such that
\[
    \left\{ \left(F_{X_1 \vert Z}(x_1 \vert z),\ldots, F_{X_{d_b} \vert 
                Z}(x_{d_b} \vert z) \right) : z \in \supp(Z \vert 
        (X_{1},\ldots,X_{d_{b}})=x) 
    \right\} \subseteq \mathcal{R}.
\]
\end{theorem}
	
\begin{proof}[\textbf{Proof of Theorem \ref{thm_att}}]
From the proof of Theorem \ref{thm_id}, $\beta(r) \equiv \Exp(B \mid R=r)$ is 
identified for all $r \in \mathcal{R}$. For notational convenience, let 
$\widetilde{X} \equiv (X_{1},\ldots,X_{d_{b}})$. By iterated expectations, the 
definition of $R$, and Proposition \ref{prop_controlvariable}, we have
\begin{align*}
    \Exp(B \mid \widetilde{X}=x)
    &= \Exp_{R \mid \widetilde{X}}[ \Exp( B \mid \widetilde{X}=x, R) \mid 
    \widetilde{X}=x] \\
    &= \Exp_{Z \mid \widetilde{X}}[ \Exp( B \mid \widetilde{X}=x, R= (F_{X_1 
        \vert Z}(x_1 \vert Z),\ldots,F_{X_{d_b} \vert Z}(x_{d_b} \vert Z)) ) 
    \mid \widetilde{X}=x] \\
    &= \Exp_{Z \mid \widetilde{X}}[ \Exp( B \mid R= (F_{X_1 \vert Z}(x_1 \vert 
    Z),\ldots,F_{X_{d_b} \vert Z}(x_{d_b} \vert Z)) )\mid \widetilde{X}=x] \\
    &= \Exp_{Z \mid \widetilde{X}}[ \beta((F_{X_1 \vert Z}(x_1 \vert 
    Z),\ldots,F_{X_{d_b} \vert Z}(x_{d_b} \vert Z))) \mid \widetilde{X}=x],
\end{align*}
which is identified since $(F_{X_1 \vert Z}(x_1 \vert z),\ldots, F_{X_{d_b} 
    \vert Z}(x_{d_b} \vert z) ) \in \mathcal{R}$ for all $z \in \supp(Z \vert 
\widetilde{X}=x)$.
\end{proof}

The support condition in Theorem \ref{thm_att} holds trivially if $\mathcal{R} = 
(0,1)$. To interpret the support condition when $\mathcal{R}$ is a strict subset 
of $(0,1)$, suppose for simplicity there is a single basic endogenous variable.  
Then the condition states that for every $z \in \supp(Z)$ such that there is an 
$r$ with $x = h[z,Q_V(r)]$, or equivalently $x = Q_{X \vert Z}(r \vert z)$, that 
$r$ must be such that we can identify $\beta(r)$. That is, the value $x$ must be 
obtainable via some $r$ for which we can identify $\beta(r)$. For example, in 
the simple model $Y = B_0 + B_1 X$, $\Var[ F_{X \mid Z}(x \mid Z) ] > 0$ is 
sufficient for the support condition. This variance condition says there are at 
least two different instrument values $z$ and $z'$ which could have yielded $x$, 
and which correspond to different conditional ranks $r$ and $r'$. That is, $x = 
h[z,Q_V(r)] = h[z', Q_V(r')]$. By strict monotonicity in the first stage 
equation, $h[z,Q_V(r)] \neq h[z,Q_V(r')]$ and $h[z', Q_V(r)] \neq h[z', 
Q_V(r')]$. Thus $h[z,Q_V(r)] \neq h[z',Q_V(r)]$ and $h[z,Q_V(r')] \neq 
h[z',Q_V(r')]$, and hence the relevance condition \ref{as_rel} holds so that $r, 
r' \in \mathcal{R}$. For example, see Figure \ref{fig_id_intuition} in which 
$x_0$ can be obtained via either $(z,r) = (1,0.6)$ or $(z',r') = (0,0.8)$. Both 
$r=0.6$ and $r' = 0.8$ are points at which $\beta(r)$ is identified. The figure 
also shows why this condition is not necessary: consider an $x$ value much 
larger than $x_1$. Such a value may be obtained only through $z=1$, and yet the 
corresponding conditional rank may be a point at which we can identify 
$\beta(r)$.

The treatment on the treated parameter $\Exp(B \mid (X_{1},\ldots,X_{d_{b}}) = 
x)$ provides one way of exploring heterogeneity in treatment effects. A truly 
constant treatment effect would yield a function $\Exp(B \mid 
(X_{1},\ldots,X_{d_{b}}) = x)$ which is constant over $x$.  An increasing 
function would show positive correlation between received treatment and the 
coefficients, while a decreasing function would show negative correlation 
between received treatment and the coefficients.  Indeed, if $\mathcal{R} = 
(0,1)^{d_{b}}$ then $\Exp(B \mid (X_{1},\ldots,X_{d_{b}}) = x)$ is identified 
for all $x \in \supp(X_{1},\ldots,X_{d_{b}})$ and hence the correlations 
$\Exp[B_{j}X_l ] = \Exp( \Exp[ B_{j} \mid X_k = x_k, k \leq d_b ] X_l )$ are 
identified for any $j$ and any $l \leq d_{b}$.

%%%%%%%%%%%%%%%%%%%%%%%%%%%%%%%%%%%%%%%%%%%%%%%%%%%%%%%%%%%%%%%%%%%%%%%%%%%%%%%%
%%%%%%%%%%%%%%%%%%%%%%%%%%%%%%%%%%%%%%%%%%%%%%%%%%%%%%%%%%%%%%%%%%%%%%%%%%%%%%%%

\section{Estimation}
\label{sec_est}

We construct estimators of $\beta(r)$ and $\beta_{\R}$ from an i.i.d.\ sample 
$\{(Y_{i}, X_{i}, Z_{i})\}_{i=1}^{n}$ using the sample analog of the expressions 
in Theorem \ref{thm_id}. We limit our focus to the case where there is one basic 
endogenous variable ($d_{b} = 1$), although there may be any number of known 
derived endogenous variables and exogenous variables $Z$. We discuss generalizations to $d_b > 1$ at the end of the section. To simplify notation, we let $X$ denote the one basic endogenous variable in both this section and the next. As a first step towards approximating the event 
that $R = r$, we construct estimates $\widehat{R}_{i}$ of $R_{i} \equiv F_{X 
    \vert Z}(X_{i} \vert Z_{i})$ for $i = 1,\ldots,n$ as
\begin{equation}
    \label{eq_rankestimation}
    \widehat{R}_{i} \equiv \widehat{F}_{X \vert Z}(X_{i} \vert Z_{i}),
\end{equation}
where $\widehat{F}_{X \vert Z}(x \vert z)$ is an estimator of $F_{X \vert Z}(x 
\vert z)$. This step of our estimation procedure is similar to those 
of \cite{imbensnewey2009e} and \cite{jun2009joe}, among others.

The asymptotic theory we develop in the next section is general enough to allow 
for many different $\sqrt{n}$--consistent estimators $\widehat{F}_{X 
    \vert Z}$. One could use a direct estimator such as the empirical 
conditional distribution function in the case that all $Z$ variables are 
discrete. Alternatively, as pointed out by 
\cite{chernozhukovfernandez-valgalichon2010e}, one can estimate $Q_{X \vert Z}(s 
\vert z)$ at several quantiles $s$ and then use the ``pre-rearrangement'' 
operator to construct an indirect estimator
\begin{equation}
    \label{eq_prerearrangement}
    \widehat{F}_{X \vert Z}(x \vert z) = \int_{0}^{1} \IndicSmall{\widehat{Q}_{X 
            \vert Z}(s \vert z) \leq x}\, ds.
\end{equation}
\cite{chernozhukovfernandez-valmelly2009wp, 
    chernozhukovfern'andez-valmelly2012e} discuss several different parametric 
direct and indirect estimators. 

For our purposes, we prefer nonparametric direct estimators (such as the 
empirical conditional distribution function) when the dimension of $Z$ is small 
and discrete, and parametric indirect estimators when there are more than a few 
covariates. The latter are easier than direct estimators to link to primitives 
under \ref{as_fs}, since, by strict monotonicity and independence, $Q_{X \vert 
    Z}(r \vert z) = h(z, Q_{V}(r))$. For example, the linear quantile regression 
model of \cite{koenkerbassett1978e} implies that $h(z,Q_{V}(r)) = Q_{X \vert 
    Z}(r \vert z) = z'\pi(r)$ for a function $\pi$ that is strictly increasing 
in $r$. Substituting $F_{V}(v)$ for $r$, we have
$h(z,v) = z'\pi(F_{V}(v))$, so that the linear quantile regression model imposes 
that $h$ is linear with respect to $z$, while \ref{as_fs} links together the 
components of $\pi$ to depend on a single underlying random variable $V$. For 
practical implementation when $Z$ has more than just a few components, we 
advocate using linear quantile regression together with 
\eqref{eq_prerearrangement} to construct $\widehat{F}_{X \vert Z}$ and 
$\widehat{R}_{i}$. Besides being easy to interpret under \ref{as_fs}, the linear 
quantile regression estimator has the additional benefits of being 
straightforward to compute, amenable to high-dimensional $Z$, and widely 
available in statistical packages. The integral in \eqref{eq_prerearrangement} 
can be evaluated using a uniform grid $\{s_{j}\}_{j=1}^{J} \subset (0,1)$.  

Having constructed $\widehat{R}_{i}$, we estimate $\beta(r)$ for a given $r$ as
\begin{equation}
    \label{eq_betahatr}
    \widehat{\beta}(r) \equiv 
    \left(\frac{1}{n}\sum_{i=1}^{n}\widehat{k}_{i}^{h}(r)W_{i}W_{i}'\right)^{+}
    \left(\frac{1}{n}\sum_{i=1}^{n}\widehat{k}_{i}^{h}(r)W_{i}Y_{i}\right),
\end{equation}
where $(\cdot)^{+}$ is the Moore-Penrose inverse and $\widehat{k}_{i}^{h}(r) 
\equiv h^{-1}K((\widehat{R}_{i} - r)/h)$ are weights constructed through a 
kernel function $K$ with bandwidth parameter $h$ that tends to $0$ 
asymptotically. The Moore-Penrose inverse is useful here because the matrix in 
question may not be invertible for all values of $r$ and $h$ in small samples, 
although our assumptions in the next section will ensure invertibility 
asymptotically. Since $R$ is always distributed uniformly with support $[0,1]$ 
when $d_{b} = 1$, we can use our estimates of $\beta(r)$ to estimate 
$\beta_{\R}$ by
\begin{equation}
    \label{eq_betahat_intr}
    \widehat{\beta}_{\R} \equiv 
    \lambda(\mathcal{R})^{-1}\int_{\R}\widehat{\beta}(r)\, dr,
\end{equation}
where $\mathcal{R}$ is a measurable subset of $[0,1]$ that is specified by the 
analyst and $\lambda$ is the Lebesgue measure.\footnote{Here and throughout the 
    paper, the integration of vectors as in \eqref{eq_betahat_intr} should be 
    understood as component-wise.} As we show in Section \ref{sec_asymptotics}, 
this estimator is $\sqrt{n}$--consistent and asymptotically normal for 
$\beta_{\R}$ under relatively weak regularity conditions. In studying a related 
problem for a different model, \cite{hoderleinsherman2013cwp4} described the 
strategy of an estimator like $\widehat{\beta}_{\R}$ as 
``localize-then-average.'' We find this terminology appealing as it captures the 
idea that for any given $r$, $\widehat{\beta}(r)$ only depends on the portion of 
the data local to the event $R = r$, while $\widehat{\beta}_{\R}$ forms an 
average of these various local estimators.

Overall, the computational complexity of \eqref{eq_betahat_intr} is very light 
for modern computing systems. A typical implementation would first estimate 
$\widehat{R}_{i}$, e.g.\ by using \eqref{eq_prerearrangement} with a moderate 
sized grid. Next, one would numerically integrate to compute 
\eqref{eq_betahat_intr}.  A simple and effective way to do this is to use 
variance-reducing pseudo-random draws, such as Halton sequences (see e.g.\  
Section 9.3.3 of \citealt{train2002}) or a uniform grid. Typically, a few hundred draws should be more than sufficient.  Moreover, unlike Monte Carlo integration, deterministic sequences can yield the same numerical results for all researchers. At each draw, one would estimate $\widehat{\beta}(r)$ using \eqref{eq_betahatr}, which 
is essentially just a weighted linear regression. Finally, the draws are averaged together to obtain $\widehat{\beta}_{\R}$. 

As we discuss in the next section, $\widehat{\beta}_{\R}$ is 
$\sqrt{n}$--consistent and asymptotically normal, but the asymptotic variance 
turns out to be complicated due to the effect of estimating $R_{i}$.  
Consequently, we use the nonparametric bootstrap to obtain standard errors. The 
typical procedure draws $S$ sets of $n$ observations with replacement from 
$\{(Y_{i}, X_{i}, Z_{i})\}_{i=1}^{n}$, say $\{(Y_{si}, X_{si}, 
    Z_{si})\}_{i=1}^{n}$ for $s = 1,\ldots,S$. These observations are used to 
compute $\widehat{\beta}^{s}_{\R}$ for $s = 1,\ldots,S$. Then
\[
    \widehat{\Sigma} \equiv \frac{1}{S-1} 
    \sum_{s=1}^{S}(\widehat{\beta}^{s}_{\R} - 
    \overline{\beta}_{\R})(\widehat{\beta}^{s}_{\R} - \overline{\beta}_{\R})'
\]
with $\overline{\beta}_{\R} \equiv S^{-1}\sum_{s=1}^{S} 
\widehat{\beta}^{s}_{\R}$ forms a bootstrapped estimate of the variance of 
$\widehat{\beta}_{\R}$. This estimator can be used to construct confidence 
intervals or conduct hypothesis tests in the usual fashion. For example, a 
two-sided confidence interval of level $\alpha$ for the first component of 
$\beta_{\R}$ would be given by \[
    \left[\widehat{\beta}_{\R, 1} - \widehat{\Sigma}_{11}^{1/2}\Phi^{-1}(1 - 
        \alpha/2), \widehat{\beta}_{\R, 1} + 
        \widehat{\Sigma}_{11}^{1/2}\Phi^{-1}(1 - \alpha/2) \right],
\]
where $\widehat{\beta}_{\R, 1}$ is the first component of 
$\widehat{\beta}_{\R}$, $\widehat{\Sigma}_{11}$ is the $(1,1)$ component of 
$\widehat{\Sigma}$, and $\Phi$ is the cumulative distribution function for the 
standard normal distribution.

Extending our estimator to the case where there are multiple basic endogenous 
variables ($d_{b} > 1$) requires a few modifications. First, we need to estimate 
$R_{ki} \equiv F_{X_{k} \vert Z}(X_{ki} \vert Z_i)$ for each $k=1,\ldots,d_{b}$.  
This can be done the same way as in the $d_{b} = 1$ case. Second, 
$\widehat{\beta}(r)$ in \eqref{eq_betahatr} needs to be modified so that the 
kernel weights are multivariate. The curse of dimensionality would accompany 
this sort of multivariate smoothing, and while $\widehat{\beta}_{\R}$ could 
still be expected to be formally $\sqrt{n}$--convergent under certain conditions 
on the kernel function, $K$, its small sample behavior will likely be quite poor 
with realistic sample sizes if $d_{b}$ is greater than $3$ or $4$. Third, when 
$d_{b} > 1$, the density of $R$ is no longer known \emph{a priori}, so that 
$\widehat{\beta}_{\R}$ could no longer be constructed by integrating as in 
\eqref{eq_betahat_intr}. A natural solution to the latter problem is to use the 
empirical measure to approximate the integral by taking
\begin{align*}
    \widehat{\beta}_{\R} = \frac{\sum_{i=1}^{n} \IndicSmall{\widehat{R}_{i} \in 
            \R}\widehat{\beta}(\widehat{R}_{i})}{\sum_{i=1}^{n} 
        \IndicSmall{\widehat{R}_{i} \in \R}}.
\end{align*}
The asymptotic analysis of this estimator involves third-order U-statistics and 
is much more complicated than that for \eqref{eq_betahat_intr}. Given this 
complication and since the case $d_{b} = 1$ is by far the most commonly 
encountered in applications, we focus our formal analysis in the next section on
$\widehat{\beta}_{\R}$ defined by \eqref{eq_betahat_intr}.

\section{Asymptotic Theory}
\label{sec_asymptotics}

In this section we discuss an asymptotic normality result for 
$\widehat{\beta}_{\R}$. The proof is in Appendix \ref{ap_proofs}. In the 
following, we let $P(r) \equiv \Exp[WW' \vert R = r]$ and use $\rightsquigarrow$ 
to denote convergence in distribution.

\begin{theorem}
    \label{thm_an}
    Under Assumptions I and E,
    \[
        \sqrt{n}(\widehat{\beta}_{\R} - \beta_{\R}) \rightsquigarrow N\left(0, 
            \lambda(\R)^{-2}\Exp[(\zeta_{1i} + \zeta_{2i})(\zeta_{1i} + 
            \zeta_{2i})'] \right),
     \]
     where
     \vspace{-5mm}
     \begin{align*}
        &\zeta_{1i} \equiv \IndicSmall{R_{i} \in 
            \mathcal{R}}P(R_{i})^{-1}W_{i}W_{i}'(B_{i} - \beta(R_{i})) \\
       &\zeta_{2i} \equiv -\Exp[\IndicSmall{R_{j} \in 
            \R}\xi_{i}(X_{j} \vert 
        Z_{j})P(R_{j})^{-1}W_{j}W_{j}'\dot{\beta}(R_{j}) \vert i] \qquad (j \neq 
        i),
    \end{align*}
    with all additional notation being defined below in Assumptions E.
\end{theorem}
%%%%%%%%%%%%%%%%%%%%%%%%%%%%%%%%%%%%%%%%%%%%%%%%%%%%%%%%%%%%%%%%%%%%%%%%%%%%%%%%
%%%%%%%%%%%%%%%%%%%%%%%%%%%%%%%%%%%%%%%%%%%%%%%%%%%%%%%%%%%%%%%%%%%%%%%%%%%%%%%%
%%%%%%%%%%%%%%%%%%%%%%%%%%%%%%%%%%%%%%%%%%%%%%%%%%%%%%%%%%%%%%%%%%%%%%%%%%%%%%%%

\begin{enumerate}[align=left, leftmargin=*, widest=E8.]
    \itshape
    \renewcommand{\theenumi}{Assumptions E}
    \renewcommand{\labelenumi}{\textbf{\theenumi.}}
    \item
    \setcounter{enumi}{0}

    \renewcommand{\theenumi}{E\arabic{enumi}}
    \renewcommand{\labelenumi}{\textbf{\theenumi.}}

    \item \label{as_iid} \textbf{(Random sample)} $(Y_{i}, X_{i}, Z_{i})$ is an 
        i.i.d.\ sample.

    \item \label{as_rset} \textbf{(Integration set)} $\R$ is a closed, 
        measurable subset of $[\delta, 1 - \delta]$ for some $\delta > 0$.

    \item \label{as_kernel} \textbf{(Kernel)} $K$ has support $[-1,1]$ and is
        twice continuously differentiable and symmetric around $0$ with 
        $\int_{-1}^{1} K(\eta) d\eta = 1$ and $\int_{-1}^{1} \eta^{2}K(\eta) 
        d\eta < \infty$.

    \item \label{as_bandwidth} \textbf{(Bandwidth)} As $n \rightarrow \infty$, 
        $\sqrt{n}h^{2} \rightarrow 0$ and $\sqrt{n}h/\log(n) \rightarrow 
        \infty$.

    \item \label{as_smoothness} \textbf{(Smoothness)} Every component of $P(r)$ 
        and $\beta(r)$ is twice continuously differentiable over $r \in \R$ with 
        first and second component-wise derivatives $\dot{P}(r), \ddot{P}(r), 
        \dot{\beta}(r), \ddot{\beta}(r)$.

    \item \label{as_moments_est} \textbf{(Existence of moments)} 
        $\Exp(\normSmall{WW'}^{4} \vert R \in \R)$ and $\Exp (\normSmall{B}^{4} 
        \vert R \in \R )$ are both finite.

    \item \label{as_rank_estimation} \textbf{(Rank estimation)} 
        $\widehat{R}_{i}$ is constructed from \eqref{eq_rankestimation} and for 
        all $(x,z) \in \mathcal{XZ}(\R) \equiv \{(x,z) : F_{X \vert Z}(x \vert 
            z) \in \R\}$,
        \begin{equation}
            \sqrt{n}(\widehat{F}_{X \vert Z}(x \vert z) - F_{X \vert Z}(x \vert 
            z)) = \frac{1}{\sqrt{n}}\sum_{i=1}^{n}\xi_{i}(x \vert z) + 
            \rho_{n}(x \vert z)
        \end{equation}
        with $\Exp[\xi_{i}(x \vert z)]$ $=$ $0$, $\Exp[\IndicSmall{R_{i} \in 
            \R}\xi_{j}(X_{i} \vert Z_{i})^{4}] < \infty$ for both $j = i$ and $j 
        \neq i$, and $\sup_{(x,z) \in \mathcal{XZ}(\R)}$ $\absSmall{\rho_{n}(x 
            \vert z)}$ $=$ $\op{1}$. Also, with probability approaching $1$, 
        $\widehat{F}_{X \vert Z}$ belongs to a class of functions $\mathcal{F}$ 
        such that $\log N(\epsilon, \mathcal{F}, \normSmall{\cdot}_{\infty}) < 
        C\epsilon^{-1/2}$ for some $C > 0$.\footnote{The notation 
            $N(\epsilon, \mathcal{F}, \normSmall{\cdot}_{\infty})$ stands for 
            the $\epsilon$--covering number of $\mathcal{F}$ under the 
            sup--norm; that is, the minimal number of 
            $\normSmall{\cdot}_{\infty}$--balls of radius $\epsilon$ that are 
            required to cover $\mathcal{F}$. Intuitively, the covering number is 
            a measure of the complexity of the class of functions $\mathcal{F}$.  
            See \cite{wellner1996} for more details.}
\end{enumerate}

The i.i.d.\ assumption \ref{as_iid} is standard for microeconometric 
applications and could in principle be extended to cover non-identical and/or 
dependent data frameworks. The assumption that $\R$ is closed and measurable in 
\ref{as_rset} is mild and of no practical significance. The additional 
restriction that $\mathcal{R}$ is a subset of $[\delta, 1-\delta]$ for some 
small $\delta > 0$ is made to avoid boundary issues. While these issues could 
potentially be addressed by using local linear weights, we have found that these 
work poorly in practice. This is perhaps not surprising in our framework since 
the density of $R$ is uniform, which is a particularly unfavorable case for 
local linear regression (see Remark 4 of \citealt{ruppertwand1994taos}).  
Additionally, as we discuss further below, \ref{as_rank_estimation} will in 
practice also require that $\mathcal{R}$ does not contain extremal ranks.  We 
therefore see \ref{as_rset} as a natural restriction given our identification 
strategy, although it does imply that we can only estimate a trimmed version of 
$\Exp(B)$. It is likely possible to adjust our estimator to allow for $\delta 
\rightarrow 0$ asymptotically, or to use a different smoothing approach that is 
less sensitive to boundary effects, such as sieves. We leave these modifications 
for future research.

The restrictions on the kernel in \ref{as_kernel} are relatively mild and allow 
for a broad range of commonly used kernels, such as the uniform or biweight 
kernel.  We rule out kernels with unbounded support such as the Gaussian kernel 
in order to apply results in the literature on kernel regression with generated 
regressors (specifically, those in \citealt{mammenrotheschienle2012taos}). Our 
bandwidth conditions in \ref{as_bandwidth} prescribe a choice of $h$ that 
undersmooths (goes to zero faster) relative to the usual optimal bandwidth 
choice for nonparametric kernel regression. This is standard given the 
semiparametric nature of our estimator $\widehat{\beta}_{\R}$ and appears in 
similar contexts like the average derivative estimator of 
\cite{powellstockstoker1989e}.  Intuitively, while $\widehat{\beta}(r)$ only 
uses a portion of the data,  $\widehat{\beta}_{\R}$ uses all the 
data and thus has a much smaller variance.  Consequently, the bandwidth $h$ can 
be sent to $0$ more quickly in order to remove the bias of 
$\widehat{\beta}_{\R}$ at a $\sqrt{n}$--rate and achieve the overall 
$\sqrt{n}$--rate of convergence asserted in Theorem \ref{thm_an}.

Assumption \ref{as_smoothness} places some standard smoothness conditions on the 
population objects $P(r)$ and $\beta(r)$. In combination with \ref{as_rel} and 
\ref{as_rset}, these imply that $P(r)$ is invertible uniformly over $\R$, and so 
serves to strengthen \ref{as_rel} in a way that is theoretically important for 
the asymptotics. The practical implication is that $\R$ should not include 
neighborhoods of isolated points where \ref{as_rel} fails, such as where the 
curves cross in Figure \ref{fig_id_intuition} ($r = .3$). Assumption 
\ref{as_moments_est} is a standard type of assumption regarding the number of 
existing moments for $W$ and $B$. Since $WW'$ contains squared terms, 
\ref{as_moments_est} essentially requires each component of $W$ to have a finite 
eighth moment.

The conditions in \ref{as_rank_estimation} require the estimator of $F_{X \vert 
    Z}$ used in constructing $\widehat{R}_{i}$ to be asymptotically linear and 
$\sqrt{n}$--convergent. This assumption is not very restrictive for parametric 
models.  
\cite{chernozhukovfernandez-valmelly2009wp,chernozhukovfern'andez-valmelly2012e} 
provide several examples of direct conditional distribution function estimators 
that satisfy this condition. In addition, 
\cite{chernozhukovfernandez-valgalichon2010e} show that 
\eqref{eq_prerearrangement}, viewed as a functional mapping from conditional 
quantile to conditional distribution functions, is Hadamard differentiable.  As 
a result, asymptotically linear representations for conditional quantile 
estimators give rise to asymptotically linear representations for 
conditional distribution estimators defined by \eqref{eq_prerearrangement} after 
applying the functional delta method. The results from the vast literature on 
quantile regression can therefore be transferred fairly easily to conditional 
distribution estimators defined by \eqref{eq_prerearrangement}. 

The more restrictive part of \ref{as_rank_estimation} is the requirement that 
the estimation error $\rho_{n}$ be convergent uniformly over $x$ and $z$ such 
that $F_{X \vert Z}(x \vert z) \in \R$. For some estimators of $F_{X \vert Z}$, 
such as the conditional empirical distribution function, this condition does not 
present a problem. For our preferred estimator that uses 
\eqref{eq_prerearrangement} and a linear quantile regression estimator of $Q_{X 
    \vert Z}$, it is well-known that this condition will generally not hold for 
subsets $\R$ that include extremal points in $[0,1]$ unless strong restrictions 
are placed on the tail behavior of $X$. However, since we rule out extremal 
ranks in \ref{as_rset}, this does not represent a substantive additional 
restriction in our setting. The following result, which is Theorem 3 in 
\cite{chernozhukovfernandezvalkowalski2011wp}, provides sufficient conditions 
for \ref{as_rank_estimation} for our preferred estimator of $\widehat{R}_{i}$.

\begin{proposition}
    \label{prop_linearqr}
    Suppose that $\widehat{F}_{X \vert Z}$ is estimated using 
    \eqref{eq_prerearrangement} with $\widehat{Q}_{X \vert Z}$ taken as the 
    linear quantile regression estimator of \cite{koenkerbassett1978e}. Then 
    \ref{as_rank_estimation} holds under Assumptions QR with
    \begin{align*}
        \xi_{i}(x \vert z) &= f_{X \vert Z}(x \vert z)z'\Exp\left[f_{X \vert 
                Z}(Z'\pi_{0}(F_{X \vert Z}(x \vert z)) \vert Z)ZZ'\right]^{-1}
        \\
        &\qquad\quad \times \left(F_{X \vert Z}(x \vert z) - \Indic{X_{i} 
                \leq \pi_{0}(F_{X \vert Z}(x \vert z))}\right)Z_{i}.
    \end{align*}
\end{proposition}
\begin{enumerate}[align=left, leftmargin=*, widest=QR1.]
    \itshape
    \renewcommand{\theenumi}{Assumptions QR}
    \renewcommand{\labelenumi}{\textbf{\theenumi.}}
    \item
    \setcounter{enumi}{0}

    \renewcommand{\theenumi}{QR\arabic{enumi}}
    \renewcommand{\labelenumi}{\textbf{\theenumi.}}

    \item \textbf{(Well-specified)} $Q_{X \vert Z}(r \vert z) = z'\pi_{0}(r)$ 
        for all $r \in \R$ and $z \in \supp(Z)$.
    \item \textbf{(Smooth quantile function)} $Q_{X \vert Z}(r \vert z)$ is 
        three times continuously differentiable in $r$ over $\mathcal{R}$ with a 
        uniformly bounded third derivative.
    \item \textbf{(Well-behaved density)} $f_{X \vert Z}(x \vert z)$ is 
        uniformly continuous, uniformly bounded and uniformly bounded away from 
        $0$ over $(x,z) \in \mathcal{XZ}(\R)$.
    \item \textbf{(Existence of moments)} $\Exp(\normSmall{Z}^{8}) < \infty$.
    \item \textbf{(No multicollinearity)} $\Exp( ZZ' )$ is invertible.
\end{enumerate}

The asymptotic variance of $\widehat{\beta}_{\R}$ given in Theorem \ref{thm_an} 
depends on two components.  If $R_{i}$ were known and did not need to be 
estimated by $\widehat{R}_{i}$, the asymptotic variance would only depend on 
$\zeta_{1i}$ and would be given by 
$\lambda(\R)^{-2}\Exp[\zeta_{1i}\zeta_{1i}']$. To interpret this quantity, 
rewrite $Y_i = W_i'B_i$ as $Y_{i} = W_{i}'\beta(r) + U_{i}(r)$, where $U_{i}(r) 
\equiv W_{i}'(B_{i} - \beta(r))$ satisfies $\Exp[U_{i}(r) \vert R_{i} = r] = 0$ 
by Proposition \ref{prop_controlvariable}.  Suppose that we were to regress $Y$ 
on $X$ in a large sample drawn from the subpopulation $R = r$.  (Of course, even 
if we knew $R_{i}$ \emph{a priori}, this wouldn't be feasible since the event $R 
= r$ has measure zero.) Then the asymptotic variance of the coefficient vector 
would be given by the usual sandwich form, 
$P(r)^{-1}\Exp[U_{i}(r)^{2}W_{i}W_{i}' \vert R_{i} = r]P(r)^{-1}$, where all of 
the typical components have been conditioned on $R = r$. This sandwich form is 
exactly what appears in
\[
    \lambda(\R)^{-2}\Exp[\zeta_{1i}\zeta_{1i}'] = 
    \Exp\left[P(R_{i})^{-1}\Exp[U_{i}(R_{i})^{2}W_{i}W_{i}' \vert 
        R_{i}]P(R_{i})^{-1} \vert R_{i} \in \R \right]\lambda(\mathcal{R})^{-1},
\]
except that it is now being integrated over all $r \in \R$ under consideration 
and scaled to account for the size of $\R$.

The second component of the asymptotic variance expression, $\zeta_{2i}$,
accounts for the effect of estimating $\widehat{R}_{i}$. This term involves the 
influence function from the first stage, $\xi_{i}$, and so will depend on the 
estimator of $F_{X \vert Z}$ that is used. It appears to generally have a 
complicated form, and at least for our preferred rank estimator discussed in 
Proposition \ref{prop_linearqr}, we have not found that the form of $\xi_{i}$ 
provides any useful simplification in the expression for $\zeta_{2i}$. Note also that $\zeta_{2i}$ depends multiplicatively on the first 
derivative of $\beta(r)$. Hence, in the case of no treatment effect 
heterogeneity, $\zeta_{2i}$ is identically zero and the asymptotic variance of 
$\widehat{\beta}_{\R}$ is determined exclusively by $\zeta_{1i}$.

Constructing a direct estimator of the asymptotic variance of 
$\widehat{\beta}_{\R}$ would be tedious and difficult, likely requiring
an additional estimator of $\dot{\beta}(r)$ as a function of $r$. Instead, we 
propose bootstrapping to approximate the limiting distribution of 
$\widehat{\beta}_{\R}$. The procedure for constructing bootstrapped standard 
errors and confidence intervals was outlined in Section \ref{sec_est}. This type 
of bootstrap procedure is generally consistent, and our framework does not 
possess any of the usual causes of inconsistency that have been studied in the 
literature. We therefore anticipate that the bootstrap is consistent, although 
we have not attempted a formal proof.

%%%%%%%%%%%%%%%%%%%%%%%%%%%%%%%%%%%%%%%%%%%%%%%%%%%%%%%%%%%%%%%%%%%%%%%%%%%%%%%%
%%%%%%%%%%%%%%%%%%%%%%%%%%%%%%%%%%%%%%%%%%%%%%%%%%%%%%%%%%%%%%%%%%%%%%%%%%%%%%%%

\section{Monte Carlo Simulations}
\label{sec_mc}
This section contains the results of Monte Carlo simulations on the 
finite-sample behavior of $\widehat{\beta}_{\R}$. We consider a data generating 
process with an outcome equation specified as
\[
    Y = B_{0} + B_{1} X,
\]
and with a first stage equation given by
\[
    X = \pi Z + \gamma Z V + V.
\]
We draw $V$ independently from a normal distribution with mean $0.1$ and standard 
deviation $0.4$. The random coefficients in the outcome equation are then 
generated as $B_{j} = \rho_{j}V + \epsilon_{j}$ for $j = 0,1$ with 
$\epsilon_{j}$ distributed $N(\mu_{j}, \sigma^{2}_{j})$ independently of all 
other variables. In particular, we take $\rho_{0} = .3$, $\mu_{0} = .2$, 
$\sigma_{0} = .2$ and $\rho_{1} = .7, \mu_{1} = .45$, $\sigma_{1} = 1$, which 
implies that $\Exp(B_{0}) = .23$ and $\Exp(B_{1}) = .52$. Since both $\rho_{j} 
\neq 0$, there is a strong endogeneity problem in this data generating process 
in the sense that $B$ and $X$ are highly correlated through their mutual 
dependence on $V$. As a consequence, the ordinary least squares (OLS) estimator 
will be inconsistent for $\Exp(B)$.

In the first stage equation we set $\pi = .2$ and consider the cases where 
$\gamma = 0$ and $\gamma = .4$. In the first case, the effect of $Z$ on $X$ is 
homogenous, so the results of \cite{heckmanvytlacil1998tjohr} and 
\cite{wooldridge1997el,wooldridge2003el,wooldridge2008ic} imply that the TSLS 
estimator will be consistent for $\Exp(B_{1})$, although it is still generally 
inconsistent for $\Exp(B_{0})$.  In the second case, the effect of $Z$ on $X$ 
varies with $V$, so that TSLS will generally be inconsistent for both 
$\Exp(B_{0})$ and $\Exp(B_{1})$. In contrast, $\widehat{\beta}_{\R}$ will be 
consistent for both components of $\Exp(B)$ for either value of $\gamma$. The 
instrument $Z$ is a binary random variable that takes values $\{0,1\}$ with 
equal probability and is drawn independently from $(V, \epsilon_{0}, 
\epsilon_{1})$. We used a conditional empirical distribution function to 
estimate $\widehat{R}_{i}$, a biweight kernel for $K$ and specified $\R = 
[0,1]$. Although this choice of $\mathcal{R}$ does not satisfy \ref{as_rset}, we 
    have found the results of these simulations to be insensitive to different 
    values of $\delta$. The number of replications in all 
simulations presented is $1000$ and the integrals in the definition of 
$\widehat{\beta}_{\R}$ were evaluated using $300$ Halton draws.

\begin{table}[H]
    \centering
    \begin{tabular}{c|ccc|ccc}
    \toprule
    % Header row
    estimators of&
    \multicolumn{3}{c|}{$N = 500$} &
    \multicolumn{3}{c}{$N = 1000$} \\

    $\Exp(B_{0}) = .23$& 
    bias & (std) & mse &
    bias & (std) & mse \\

    \hline

        OLS 
        
           &

        0.0291 
        
           &

        (0.0228) 
        
           &

        0.0014 
        
           &

        0.0294 
        
           &

        (0.0172) 
        
           &

        0.0012

    \\

        TSLS 
        
           &

        0.1124 
        
           &

        (0.0581) 
        
           &

        0.0160 
        
           &

        0.1105 
        
           &

        (0.0393) 
        
           &

        0.0137

    \\

    \hline
    \hline

            h = 0.01
             
               &

            -0.0136
             
               &

            (0.1784)
             
               &

            0.0320
             
               &

            -0.0164
             
               &

            (0.1085)
             
               &

            0.0120

        \\

            h = 0.03
             
               &

            -0.0394
             
               &

            (0.1322)
             
               &

            0.0190
             
               &

            -0.0391
             
               &

            (0.0830)
             
               &

            0.0084

        \\

            h = 0.05
             
               &

            -0.0590
             
               &

            (0.1118)
             
               &

            0.0160
             
               &

            -0.0571
             
               &

            (0.0703)
             
               &

            0.0082

        \\

            h = 0.07
             
               &

            -0.0750
             
               &

            (0.0988)
             
               &

            0.0154
             
               &

            -0.0724
             
               &

            (0.0620)
             
               &

            0.0091

        \\

            h = 0.09
             
               &

            -0.0878
             
               &

            (0.0896)
             
               &

            0.0157
             
               &

            -0.0852
             
               &

            (0.0558)
             
               &

            0.0104

        \\

            h = 0.11
             
               &

            -0.0980
             
               &

            (0.0822)
             
               &

            0.0164
             
               &

            -0.0957
             
               &

            (0.0510)
             
               &

            0.0118

        \\

            h = 0.13
             
               &

            -0.1066
             
               &

            (0.0762)
             
               &

            0.0172
             
               &

            -0.1045
             
               &

            (0.0472)
             
               &

            0.0132

        \\

            h = 0.15
             
               &

            -0.1137
             
               &

            (0.0716)
             
               &

            0.0181
             
               &

            -0.1120
             
               &

            (0.0444)
             
               &

            0.0145

        \\

\toprule
\toprule

    % Header row
    estimators of&
    \multicolumn{3}{c|}{$N = 500$} &
    \multicolumn{3}{c}{$N = 1000$} \\

    $\Exp(B_{1}) = .52$& 
    bias & (std) & mse &
    bias & (std) & mse \\

    \hline

        OLS 
        
           &

        0.4136 
        
           &

        (0.0889) 
        
           &

        0.1790 
        
           &

        0.4142 
        
           &

        (0.0648) 
        
           &

        0.1757

    \\

        TSLS 
        
           &

        -0.0024 
        
           &

        (0.2757) 
        
           &

        0.0760 
        
           &

        0.0104 
        
           &

        (0.1877) 
        
           &

        0.0353

    \\

    \hline
    \hline

            h = 0.01
             
               &

            -0.0125
             
               &

            (0.2668)
             
               &

            0.0713
             
               &

            0.0094
             
               &

            (0.1700)
             
               &

            0.0290

        \\

            h = 0.03
             
               &

            0.0207
             
               &

            (0.2256)
             
               &

            0.0513
             
               &

            0.0338
             
               &

            (0.1491)
             
               &

            0.0234

        \\

            h = 0.05
             
               &

            0.0532
             
               &

            (0.2012)
             
               &

            0.0433
             
               &

            0.0621
             
               &

            (0.1333)
             
               &

            0.0216

        \\

            h = 0.07
             
               &

            0.0853
             
               &

            (0.1826)
             
               &

            0.0406
             
               &

            0.0908
             
               &

            (0.1211)
             
               &

            0.0229

        \\

            h = 0.09
             
               &

            0.1158
             
               &

            (0.1674)
             
               &

            0.0414
             
               &

            0.1192
             
               &

            (0.1111)
             
               &

            0.0265

        \\

            h = 0.11
             
               &

            0.1442
             
               &

            (0.1549)
             
               &

            0.0448
             
               &

            0.1463
             
               &

            (0.1028)
             
               &

            0.0320

        \\

            h = 0.13
             
               &

            0.1706
             
               &

            (0.1447)
             
               &

            0.0501
             
               &

            0.1719
             
               &

            (0.0963)
             
               &

            0.0388

        \\

            h = 0.15
             
               &

            0.1948
             
               &

            (0.1368)
             
               &

            0.0566
             
               &

            0.1955
             
               &

            (0.0914)
             
               &

            0.0466

        \\
    
    \toprule
\end{tabular}

    \caption{Performance of $\widehat{\beta}_{\R}$ as an estimator of $\Exp(B)$ 
        relative to ordinary least squares (OLS) and two stage least squares 
        (TSLS) in the dgp without first stage heterogeneity ($\gamma = 0$).}
    \label{tbl_gamma0}
\end{table}

Table \ref{tbl_gamma0} reports the performance of the first and second 
components of $\widehat{\beta}_{\R}$ as estimators of $\Exp(B_{0})$ and 
$\Exp(B_{1})$ relative to both the OLS and TSLS estimators in the case without first stage heterogeneity, $\gamma = 0$. As expected, the OLS estimator is inconsistent for both 
parameters. The results support the prediction of
\cite{heckmanvytlacil1998tjohr} and 
\cite{wooldridge1997el,wooldridge2003el,wooldridge2008ic} that TSLS is 
consistent for $\Exp(B_{1})$ but inconsistent for $\Exp(B_{0})$.  The 
performance of $\widehat{\beta}_{\R}$ is reported for a variety of bandwidth 
choices $h$. Our prediction of the consistency of both components of 
$\widehat{\beta}_{\R}$ at the $\sqrt{n}$ rate is supported by the decrease in 
mean squared error (mse) that occurs when increasing $N$ from $500$ to $1000$.  
Most remarkable is the performance of the second component of 
$\widehat{\beta}_{\R}$ as an estimator of $\Exp(B_{1})$ relative to TSLS. Our 
results suggest a mean-squared error that is actually slightly lower than TSLS 
across a broad range of bandwidth values. For smaller bandwidth values, both the 
bias and standard deviation (std) are comparable to TSLS, perhaps even being a 
bit smaller.

\begin{table}[H]
    \centering
    \begin{tabular}{c|ccc|ccc}
    \toprule
    % Header row
    estimators of&
    \multicolumn{3}{c|}{$N = 500$} &
    \multicolumn{3}{c}{$N = 1000$} \\

    $\Exp(B_{0}) = .23$& 
    bias & (std) & mse &
    bias & (std) & mse \\

    \hline

        OLS 
        
           &

        0.0532 
        
           &

        (0.0254) 
        
           &

        0.0035 
        
           &

        0.0537 
        
           &

        (0.0190) 
        
           &

        0.0032

    \\

        TSLS 
        
           &

        0.0927 
        
           &

        (0.0542) 
        
           &

        0.0115 
        
           &

        0.0912 
        
           &

        (0.0360) 
        
           &

        0.0096

    \\

    \hline
    \hline

            h = 0.01
             
               &

            -0.0067
             
               &

            (0.1562)
             
               &

            0.0244
             
               &

            -0.0157
             
               &

            (0.1146)
             
               &

            0.0134

        \\

            h = 0.03
             
               &

            -0.0226
             
               &

            (0.0993)
             
               &

            0.0104
             
               &

            -0.0228
             
               &

            (0.0684)
             
               &

            0.0052

        \\

            h = 0.05
             
               &

            -0.0349
             
               &

            (0.0887)
             
               &

            0.0091
             
               &

            -0.0331
             
               &

            (0.0595)
             
               &

            0.0046

        \\

            h = 0.07
             
               &

            -0.0451
             
               &

            (0.0836)
             
               &

            0.0090
             
               &

            -0.0432
             
               &

            (0.0549)
             
               &

            0.0049

        \\

            h = 0.09
             
               &

            -0.0543
             
               &

            (0.0797)
             
               &

            0.0093
             
               &

            -0.0528
             
               &

            (0.0516)
             
               &

            0.0055

        \\

            h = 0.11
             
               &

            -0.0626
             
               &

            (0.0761)
             
               &

            0.0097
             
               &

            -0.0615
             
               &

            (0.0486)
             
               &

            0.0061

        \\

            h = 0.13
             
               &

            -0.0704
             
               &

            (0.0729)
             
               &

            0.0103
             
               &

            -0.0696
             
               &

            (0.0461)
             
               &

            0.0070

        \\

            h = 0.15
             
               &

            -0.0776
             
               &

            (0.0702)
             
               &

            0.0109
             
               &

            -0.0771
             
               &

            (0.0439)
             
               &

            0.0079

        \\

    \toprule
    \toprule

    % Header row
    estimators of&
    \multicolumn{3}{c|}{$N = 500$} &
    \multicolumn{3}{c}{$N = 1000$} \\

    $\Exp(B_{1}) = .52$& 
    bias & (std) & mse &
    bias & (std) & mse \\

    \hline

        OLS 
        
           &

        0.3678 
        
           &

        (0.0926) 
        
           &

        0.1439 
        
           &

        0.3690 
        
           &

        (0.0665) 
        
           &

        0.1406

    \\

        TSLS 
        
           &

        0.1888 
        
           &

        (0.2578) 
        
           &

        0.1021 
        
           &

        0.2002 
        
           &

        (0.1742) 
        
           &

        0.0704

    \\

    \hline
    \hline

            h = 0.01
             
               &

            -0.0194
             
               &

            (0.3205)
             
               &

            0.1031
             
               &

            -0.0105
             
               &

            (0.2315)
             
               &

            0.0537

        \\

            h = 0.03
             
               &

            -0.0077
             
               &

            (0.1970)
             
               &

            0.0389
             
               &

            0.0053
             
               &

            (0.1310)
             
               &

            0.0172

        \\

            h = 0.05
             
               &

            0.0101
             
               &

            (0.1655)
             
               &

            0.0275
             
               &

            0.0213
             
               &

            (0.1100)
             
               &

            0.0126

        \\

            h = 0.07
             
               &

            0.0304
             
               &

            (0.1512)
             
               &

            0.0238
             
               &

            0.0389
             
               &

            (0.1009)
             
               &

            0.0117

        \\

            h = 0.09
             
               &

            0.0515
             
               &

            (0.1423)
             
               &

            0.0229
             
               &

            0.0579
             
               &

            (0.0949)
             
               &

            0.0124

        \\

            h = 0.11
             
               &

            0.0724
             
               &

            (0.1354)
             
               &

            0.0236
             
               &

            0.0769
             
               &

            (0.0901)
             
               &

            0.0140

        \\

            h = 0.13
             
               &

            0.0923
             
               &

            (0.1298)
             
               &

            0.0254
             
               &

            0.0956
             
               &

            (0.0864)
             
               &

            0.0166

        \\

            h = 0.15
             
               &

            0.1113
             
               &

            (0.1253)
             
               &

            0.0281
             
               &

            0.1138
             
               &

            (0.0836)
             
               &

            0.0199

        \\
    
    \toprule
\end{tabular}

    \caption{Performance of $\widehat{\beta}_{\R}$ as an estimator of $\Exp(B)$ 
        relative to ordinary least squares (OLS) and two stage least squares 
        (TSLS) in the dgp with first stage heterogeneity ($\gamma = 0.4$).}
    \label{tbl_gamma.4}
\end{table}

Table \ref{tbl_gamma.4} reports the same type of results as Table 
\ref{tbl_gamma0} for the case with first stage heterogeneity, $\gamma = .4$. Here, we see that the 
heterogeneity in the effect of $Z$ on $X$ leads to severe inconsistency for the 
TSLS estimator. On the other hand, $\widehat{\beta}_{\R}$ remains consistent and 
performs similarly to the case where $\gamma = 0$. We interpret these results as 
promising evidence in support of the practical applicability of our estimator to 
situations where heterogeneity in the first stage cannot be ruled out.

%%%%%%%%%%%%%%%%%%%%%%%%%%%%%%%%%%%%%%%%%%%%%%%%%%%%%%%%%%%%%%%%%%%%%%%%%%%%%%%%
%%%%%%%%%%%%%%%%%%%%%%%%%%%%%%%%%%%%%%%%%%%%%%%%%%%%%%%%%%%%%%%%%%%%%%%%%%%%%%%%

\section{Air Pollution and House Prices}
\label{sec_app}

In this section, we apply our results to analyze the relationship between air 
pollution and house prices. A large literature on hedonic methods uses 
relationships like this to infer the value of non-market amenities, such as 
clean air (e.g. \citealt{Rosen1974}, \citealt{SmithHuang1995}, 
\citealt{EkelandHeckmanNesheim2004}, \citealt{Palmquist2005}, 
\citealt{HeckmanMatzkinNesheim2010}). Reliable measurements of these valuations 
are important for quantifying the economic benefit of air quality regulation. We 
follow the empirical approach of \cite{chaygreenstone2005jope}. They argue that 
previous analysis based on cross-sectional OLS or first-differences yields 
small, zero, or perverse-signed effects due to omitted variables, such as 
unobserved economic shocks, or sorting of households based on unobserved 
preferences for clean air. To remedy this, they use regulation introduced by the 1970 Clean Air Act 
Amendments to define a binary instrument for change in total suspended 
particulates (TSP) from 1970 to 1980 and then use TSLS to estimate the effect of 
TSP changes on county-level house price changes.

Our analysis builds on Chay and Greenstone in several directions. As they note 
(page 393), a correlated random coefficients model is appropriate due to sorting 
of households based on unobserved preferences for clean air (we also discuss 
this below). We demonstrate substantial first stage heterogeneity in the effect 
of the instrument, which strongly suggests that the simpler estimators discussed by 
\cite{heckmanvytlacil1998tjohr} and 
\cite{wooldridge1997el,wooldridge2003el,wooldridge2008ic} would be inconsistent for the APE. Likewise, the binary instrument precludes approaches which rely on continuous variation, such as \cite{florensheckmanmeghiretal2008e}. For two subsets of counties where the instrument has a 
statistically significant effect on pollution levels, we estimate unweighted 
average partial effects of changes in pollution on changes in house prices. 
These estimates demonstrate patterns that are consistent with household sorting. 
Taken together, these estimates along with TSLS suggest there is substantial 
heterogeneity in households' value of clean air.

\subsection{The dataset and institutional background}

\nocite{CCDB1977}
\nocite{CCDB1983}
\nocite{TSPdata}

The 1970 Clean Air Act Amendments set national ambient air quality standards 
(NAAQS) for TSPs with the goal that all counties would eventually meet these 
standards. The law requires the U.S.\ Environmental Protection Agency (EPA) to annually designate each county as either 
attainment, if the county meets the standard, or as nonattaintment if it does 
not. Firms in nonattainment counties are subject to much stricter pollution 
regulations than firms in attainment counties. Consequently, nonattainment 
status should affect counties' pollution levels. Chay and Greenstone argue 
(pages 395--406) that a county's nonattainment status in 1975 and 1976 is 
plausibly independent of unobserved variables which change between 1970 and 1980 
and affect housing prices, such as unobserved economic shocks, as well as 
unobserved changes in clean air preferences from 1970 to 1980 due to sorting.  
In addition, there is no reason to expect that households care about 
nonattainment status above and beyond its effect on pollution; i.e., the 
exclusion restriction holds. For these reasons, Chay and Greenstone conclude 
that mid-decade nonattaintment status is a valid instrument for identifying 
causal effects of changes in TSP from 1970 to 1980 on changes in house prices 
from 1970 to 1980. We take the instrument definition and validity arguments as given and investigate the implications of allowing for first stage heterogeneity via our CRC estimator.

Our dataset is essentially identical to that of \cite{chaygreenstone2005jope}, 
as described in their data appendix (pages 419--421). We obtain house price data 
as well as covariates from the 1972 and 1983 County and City Data Books 
(obtained via ICPSR). This price and covariate data is only available at the 
county level and hence the units of analysis are counties. TSP pollution data 
may be downloaded from the EPA. One minor difference between our dataset and 
Chay and Greenstone's is that we do not have TSP data from 1969; this data is 
not available for download and the EPA has not responded to our requests. Chay 
and Greenstone define TSP levels for 1970 as the average of TSP levels for 
1969--1972. Since we are missing 1969, we average over just 1970--1972. To 
define the 1980 TSP level, we average over 1977--1980 levels, as in Chay and 
Greenstone. The annual TSP levels are derived by aggregating observations throughout the year at 
different pollution monitors located across the country, as in Chay 
and Greenstone (page 384). Also as in their paper (page 391), we use the TSP 
data from 1974 and 1975 to define the instrument as the binary indicator 
variable for mid-decade nonattaintment status, since data on the actual EPA 
designated nonattaintment status does not exist. A second difference between 
their paper and our analysis arises here. Of our 989 observations, our 
definition of the instrument yielded 300 nonattaintment counties, whereas Chay 
and Greenstone have only 280 nonattaintment counties out of 988 
observations.\footnote{This difference may arise due to an ambiguity in 
    determining whether a county violates the ``bad day rule'', which says that 
    the second largest daily TSP value within a year must not exceed 260 
    $\mu$g/m$^3$ and would place a county in nonattainment. For counties with multiple monitors, there are at least two 
    approaches: (1) compute the second highest daily TSP value for each monitor, 
    and say a county violates the rule if any monitor within the county violates 
    that rule, and (2) compute a county-level daily reading by averaging all 
    monitors for a given day, and then compute the second highest daily TSP from 
    that averaging. Our reading of the EPA regulations suggest that (1) is the 
    approach EPA used and hence is what we use as well. Approach (2) leads to 
    far fewer counties being designated as nonattaintment---222 out of 989. 
    Hence this approach cannot be what Chay and Greenstone used either. We are 
    unsure what they used, and neither author has responded to our requests for 
    clarification.} This difference may explain why our TSLS results in table 
\ref{EmpiricalResultsTable} differ somewhat from theirs.

Table \ref{summarystatstable} shows summary statistics along with a list of all 
covariates included in the analysis (see pages 420--421 of Chay and Greenstone 
for further explanation of the covariates). This table is comparable to Chay and 
Greenstone's table 1. All prices are adjusted to 1982--1984 dollars. Mean house prices increased from 
around \$40,000 to around \$53,000 while TSP levels fell by about 9 
$\mu$g/m$^3$. The goal of the instrumental variable analysis is to determine to what extent this correlation between the rise in house prices and fall in pollution levels reflects causal effects.

\begin{table}[h]
\centering
\caption{Summary Statistics, 1970 and 1980 \label{summarystatstable}}

\setlength{\linewidth}{.1cm} 
\newcommand{\contents}{ 

\begin{tabular}{l*{2}{c}} \hline\hline
                    &        1970&        1980\\
\hline
Mean housing value  &     40,268 &     53,046 \\
Mean TSPs           &       65.5 &       56.3 \\
Income per capita   &      7,530 &      9,279 \\
Total population    &163,880,811 &175,516,811 \\
Unemployment rate   &      .0455 &       .068 \\
\% employment in manufacturing&       .249 &       .226 \\
Population density  &        613 &        476 \\
\% $\geq$ high school graduate&       .504 &       .646 \\
\% $\geq$ college graduate&      .0971 &       .146 \\
\% urban            &       .576 &       .593 \\
\% poverty          &       .124 &      .0976 \\
\% white            &       .902 &       .877 \\
\% female           &        .51 &       .511 \\
\% senior citizens  &      .0997 &       .113 \\
\% overall vacancy rate&      .0336 &      .0782 \\
\% owner-occupied   &       .676 &       .638 \\
\% of houses without plumbing&      .0822 &      .0253 \\
Per capita government revenue&        748 &      1,138 \\
Per capita property taxes&        314 &        366 \\
Per capita general expenditures&        769 &      1,111 \\
\% spending on education&       .549 &       .509 \\
\% spending on highways&      .0909 &      .0698 \\
\% spending on welfare&      .0462 &      .0371 \\
\% spending on health&      .0486 &      .0669 \\
Observations        &        989 &        989 \\
\hline\hline
\multicolumn{3}{p{\linewidth}}{\footnotesize Statistics are based on the 989 counties with data on TSP in 1970, 1980, and 1974 or 1975, as well as nonmissing price data in both 1970 and 1980. Mean TSP for 1970 is the average of 1970 to 1972 annual TSP. Mean TSP for 1980 is the average of 1977 to 1980 annual TSP. Annual TSP for a county is the weighted average of the geometric mean of each monitor's TSP readings in the county, using the number of observations per monitor as weights. All dollar quantities are adjusted to 1982-1984 dollars (housing values use the housing only part of the CPI, series CUUR0000SAH; all other values use overall CPI, series CUUR0000SA0).}\\
\end{tabular}
}
\setbox0=\hbox{\contents} 
\setlength{\linewidth}{\wd0} 
\contents
\end{table}

\subsection{Empirical results}

Let $X$ denote the change in TSP between 1970 and 1980, $Z_2$ denote mid-decade 
nonattainment status, and $Z_1$ denote the vector of 22 covariates. As discussed 
in Section \ref{sec_est}, we begin by estimating a linear quantile regression of 
the treatment variable $X$ on the instrument and the covariates for several 
different quantiles (similar results obtain when the covariates are omitted). 
Figure \ref{fig_empirical_first_stage} plots the coefficient on the instrument 
$Z_2$ against the quantile used. Recall that the first stage assumption \ref{as_fs} implies that countries with small conditional ranks generally have smaller values of $X$---that is, larger drops in pollution---than counties with larger conditional ranks. A $r$th-quantile regression tells us the effect of the instrument for counties with conditional rank $R = r$. For example, the median quantile regression tells us the effect of the instrument for counties generally at the middle of the distribution of changes in TSP. For these counties, the coefficient is around $-0.1$, which suggests that being in a nonattaintment county caused pollution to drop by $-10$ $\mu$g/m$^3$ relative to 
attainment counties, all else equal. Recall from table \ref{summarystatstable} 
that the average TSP level in 1970 was 65.5 $\mu$g/m$^3$, and it fell by around 
9 $\mu$g/m$^3$. So a $-10$ $\mu$g/m$^3$ effect is quite large. The effect of 
$-20$ $\mu$g/m$^3$ for counties with the smallest conditional ranks is even larger. Note that the 
coefficient we find at the median, about $-10$ $\mu$g/m$^3$, is essentially 
equal to the coefficient obtained from a linear mean regression, as in Chay and 
Greenstone's table 4 panel A column 2.

\begin{figure}[h]
\centering
\includegraphics[width=130mm]{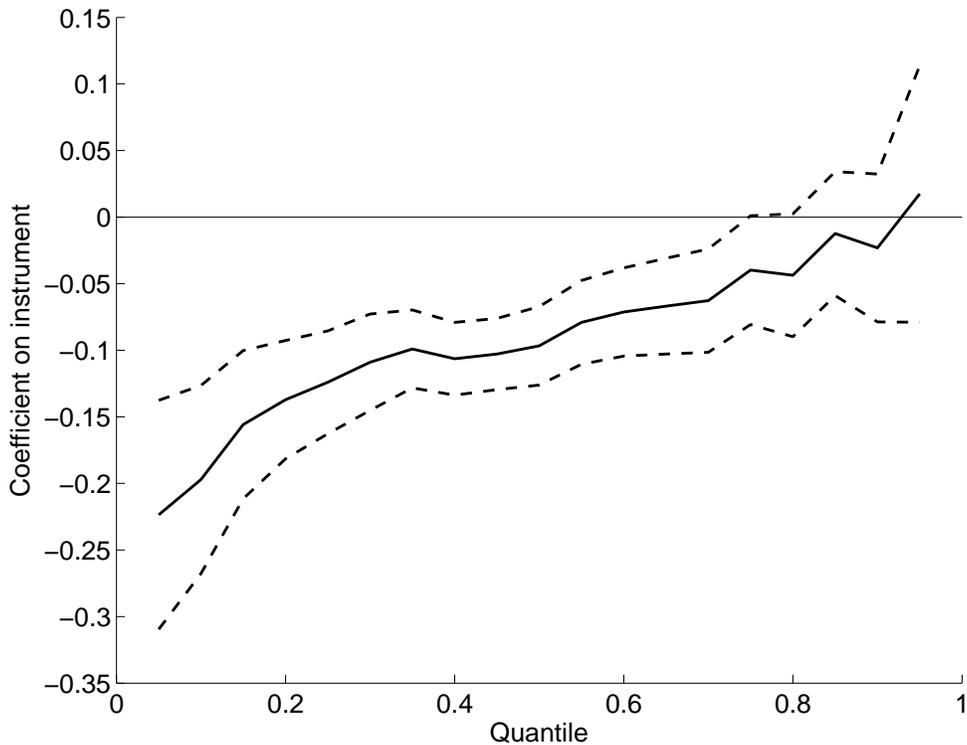}
\caption{\label{fig_empirical_first_stage} Plot of estimation results from 
    several linear quantile regressions of treatment $X$ (change in TSP, where 
    TSP here is in units of $1 \times 10^{-4}$ grams/m$^3$ rather than $1 \times 
    10^{-6}$ grams/m$^3 = \mu$g/m$^3$) on the instrument $Z_2$ (nonattainment 
    status) and controls. The solid line plots the estimated coefficient on the 
    instrument on the vertical axis against the corresponding quantile on the 
    horizontal axis, from $.05$ to $.95$ in $.05$ increments. The dotted lines 
    plot simultaneous confidence intervals for each of these quantiles.}
\end{figure}

The plot shows a heterogeneous effect of the instrument on treatment---the 
instrument has a strong negative effect at low quantiles, but this effect 
decreases towards zero for higher quantiles. For quantiles from about 0.75 to 1 
the instrument does not have a statistically significant effect on treatment. 
Even though the standard $F$-statistic for the instrument suggests that there is 
no weak instrument problem ($F = 25$), figure \ref{fig_empirical_first_stage} 
shows that the instrument is not uniformly strong for all counties, and indeed 
there are many counties (about 25\% of them) where the instrument appears to have no effect at all.

Counties with smaller conditional ranks have the smallest values of $X$---that 
is, the largest drops in pollution over the decade. Large drops in pollution are 
strongly negatively correlated with having a high baseline level of pollution in 
1970 ($\rho = -0.78$). Consequently, the heterogeneous effect of the instrument 
is to be expected. Counties in nonattainment are more heavily regulated than 
counties in attainment. But being regulated only matters if a county has a 
pollution problem to begin with. Hence the counties with the highest baseline 
pollution are also the ones where the instrument has a strongest effect. 
Conversely, the counties with the lowest baseline pollution (and hence lowest 
potential drops in pollution) are the ones where the instrument has essentially 
a zero effect.

Next we implement our generalized CRC estimator. We choose two different sets of 
conditional ranks to consider: $\mathcal{R} = [0.1,0.4]$ and $\mathcal{R} = 
[0.4,0.7]$. First, we omit estimating average partial effects near the tails as 
discussed in Section \ref{sec_asymptotics}. Second, we omit estimating average 
partial effects near the region where the instrument is weak or irrelevant, as 
discussed in Sections \ref{sec_model} and \ref{sec_asymptotics}. We split up the 
region $[0.1,0.7]$ for which we can estimate average partial effects into two 
pieces in order to examine potential heterogeneity in the effect of pollution on 
house prices. For each choice of $\mathcal{R}$ we use the following tuning 
parameters: 1999 points for the first stage grid (equation 
\ref{eq_prerearrangement}), which corresponds to 1999 linear quantile 
regressions (an equally spaced grid with step size $0.0005$). 2000 Halton 
draws for integration of $\widehat{\beta}(r)$ over $r \in \mathcal{R}$ (equation 
\ref{eq_betahat_intr}). We use 500 bootstrap draws to compute 95\% confidence 
intervals. Finally, we present a range of bandwidths from $h = 0.04$ to 
$h=0.085$. In our Monte Carlo simulations, the bandwidth $h= 0.07$ minimized MSE 
for the sample size $N=1000$ and when there was first stage heterogeneity (Table 
\ref{tbl_gamma.4}).

\begin{sidewaystable}[h]
%\begin{table}[h]
\centering\caption{Estimates of the effect of 1970--1980 changes in TSP pollution on changes in log housing values \label{EmpiricalResultsTable}} 
\setlength{\linewidth}{.1cm} 
\newcommand{\contents}{ 
\begin{tabular}{c c c c c} \hline\hline 
Estimator  & & & & \\ \hline 
OLS &   \multicolumn{2}{c}{0.0861} &   \multicolumn{2}{c}{0.0288} \\
  & \multicolumn{2}{c}{[  0.0079,   0.1643]} & \multicolumn{2}{c}{[ -0.0200,  0.0776]} \\ [0.5em]
 TSLS &  \multicolumn{2}{c}{-0.4149} &  \multicolumn{2}{c}{-0.2073} \\
  & \multicolumn{2}{c}{[ -0.7616,  -0.0682]} & \multicolumn{2}{c}{[ -0.4258,  0.0113]} \\ [0.5em]
 \hline \hline
  & (1) & (2) & (3) & (4) \\ [0.5em]
Generalized CRC estimator  & $\mathcal{R} = [0.1,0.4]$ & $\mathcal{R} = [0.4,0.7]$ & $\mathcal{R} = [0.1,0.4]$ &$\mathcal{R} = [0.4,0.7]$\\ [0.5em]
  \hline
$h=0.040$ &  -0.0241 &  -0.2067 &  -0.0664 &  -0.1517\\
  & [ -0.3011,   0.2314] & [ -0.7879,   0.4544] & [ -0.4395,  0.1585] & [ -0.7414,  0.4363] \\ [0.5em]
$h=0.0475$ &  -0.0252 &  -0.1766 &  -0.0640 &  -0.1535 \\
  & [ -0.2994,   0.2239] & [ -0.7635,   0.4540] & [ -0.3981,  0.1407] & [ -0.6894,  0.3974] \\ [0.5em]
$h=0.055$ &  -0.0261 &  -0.1545 &  -0.0670 &  -0.1580 \\
  & [ -0.2974,   0.2210] & [ -0.7419,   0.4498] & [ -0.3864,  0.1215] & [ -0.6441,  0.3634]\\ [0.5em]
$h=0.0625$ &  -0.0266 &  -0.1364 &  -0.0703 &  -0.1592 \\
  & [ -0.2940,   0.2187] & [ -0.6967,   0.4410] & [ -0.3719,  0.1179] & [ -0.6258,  0.3310]\\ [0.5em]
$h=0.0775$ &  -0.0258 &  -0.1073 &  -0.0755 &  -0.1544 \\
  & [ -0.2856,   0.2047] & [ -0.6551,   0.4255] & [ -0.3422,  0.1057] & [ -0.5988,  0.2751] \\ [0.5em]
$h=0.085$ &  -0.0248 &  -0.0954 &  -0.0784 &  -0.1524 \\
  & [ -0.2807,   0.2009] & [ -0.6329,   0.4430] & [ -0.3321,  0.1031] & [ -0.5927,  0.2636]\\ [0.5em]
 \hline 
Observations &      \multicolumn{2}{c}{983} &      \multicolumn{2}{c}{983}   \\ 
  County data book controls? & \multicolumn{2}{c}{No} & \multicolumn{2}{c}{Yes} \\ 
 \hline\hline \multicolumn{5}{p{\linewidth}}{\footnotesize Entries show estimates of coefficients on change in TSP over 1970-1980, and corresponding 95-percent confidence intervals for several different estimators: ordinary least squares, two-stage least squares, and a variety of bandwidths for our generalized correlated random coefficient model estimator. TSP here is in units of $1 \times 10^{-4}$ grams/m$^3$ rather than $1 \times 10^{-6}$ grams/m$^3 = \mu$g/m$^3$. Columns (1) and (3) show the generalized CRC estimates for conditional ranks over $\mathcal{R} = [0.1,0.4]$ while columns (2) and (4) use $\mathcal{R} = [0.4,0.7]$. All regressions are first differenced from 1970 to 1980. The outcome variable is 1980 log-housing value minus 1970 log-housing value. The treatment variable of interest is the 1980 TSP value minus the 1970 TSP value. All controls are also first differenced. The instrument is mid-decade nonattainment status (see body text for further details). OLS and TSLS confidence intervals are computed via asymptotic plug-in estimators of the heteroskedasticity robust standard errors. The generalized CRC model confidence intervals are computed using 500 bootstrap draws.} \\ 
\end{tabular}
} 
\setbox0=\hbox{\contents} 
\setlength{\linewidth}{\wd0} 
\contents 
%\end{table}
\end{sidewaystable}

Table \ref{EmpiricalResultsTable} shows the main results. There are four 
columns. Columns (1) and (2) show estimation results without any control 
variables while columns (3) and (4) show estimation results with the control 
variables. Chay and Greenstone present additional specifications which use a 
``flexible functional form'', but they do not specify what precisely they mean, 
and they have not responded to our requests for clarification. Hence we present 
only their first two specifications. Columns (1) and (3) show the generalized 
CRC estimates for the choice $\mathcal{R} = [0.1,0.4]$ while columns (2) and (4) 
show estimates for $\mathcal{R} = [0.4,0.7]$.

First, notice that OLS---which is a first-differences regression here---has the 
perverse sign (implying prices go up when pollution goes up), and is even statistically significant without covariates, as also 
found by Chay and Greenstone (their Table 3 panel C columns 1 and 2).  The TSLS 
results also mirror the main findings of Chay and Greenstone (their Table 5 
panel A columns 1 and 2): Without covariates, we find a large negative effect of 
pollution on log house prices: a 1 $\mu$g/m$^3$ reduction in mean TSP causes a 
0.42 percent increase in property values. With covariates, this effect size is 
cut in half and becomes marginally statistically insignificant. Chay and 
Greenstone's corresponding TSLS point estimate is $-0.213$ with a 95 percent 
confidence interval of $[-0.4,-0.025]$, which is not too different from the 
confidence interval we obtain (the difference is likely due to the data 
discrepancies previously mentioned).

Second, consider the bootstrap confidence intervals for the generalized CRC 
estimates. None of them are statistically significant. We have already seen that the instrument is strongest for 
smaller conditional ranks. Consequently, the confidence intervals are nearly 
twice as wide for the `weaker' instrument region $\mathcal{R} = [0.4,0.7]$ than 
for the `stronger' instrument region $\mathcal{R} = [0.1,0.4]$. For the 
region $\mathcal{R} = [0.1,0.4]$, the confidence intervals are roughly the same 
length as for TSLS for column (3), and they're actually smaller for column (1).  
But the point estimates for both of these columns are close to zero, as expected 
under the sorting story discussed below. Consequently, since the TSLS confidence 
interval was already quite close to zero, the generalized CRC estimate 
confidence intervals here have just been shifted over to be centered near 
zero.

Next consider the generalized CRC estimates. Begin by considering the point 
estimates. The point estimates are fairly insensitive to changes in bandwidths, 
especially after controlling for covariates. When controlling for covariates, 
the estimates for $\mathcal{R} = [0.4,0.7]$ are roughly twice as large as those 
for $\mathcal{R} = [0.1,0.4]$. A similar finding is true when comparing column 
(2) to column (1), although the difference in magnitudes is much larger without 
the controls. Moreover, all the CRC estimates are smaller than the TSLS 
estimates. These findings are consistent with the possibility that households 
sort according to their unobserved preference for clean air during the baseline year of 1970. 
In that case, households with the strongest taste for clean air will move to 
counties with low baseline pollution, while households who do not care about 
clean air will move to counties with high baseline pollution. As we saw earlier, 
the baseline pollution and conditional rank are strongly correlated. Hence 
counties where the instrument has the strongest effect are also counties where 
most households do not care about clean air, and so we find close to zero 
effects for $\mathcal{R} = [0.1,0.4]$. For counties with households that 
moderately care about clean air, $\mathcal{R} = [0.4,0.7]$, we find moderate 
sized effects. For counties with households that care strongly about clean air, 
$\mathcal{R} = [0.7,0.1]$, we are unable to identify their preferences, since 
those are precisely the counties where the instrument has little or no 
effect---because those are the counties with little pollution to begin with. The 
finding that TSLS is larger than all the CRC estimates suggests that the effects 
for counties with $\mathcal{R} = [0.7,1]$ are larger than $-0.15$, the effect we 
found for $\mathcal{R} = [0.4,0.7]$. This is because TSLS is a weighted average 
effect, where the weights depend on the strength of the instrument. Although 
counties in $\mathcal{R} = [0.7,1]$ will receive close to zero weight in forming 
the TSLS estimand, if their actual effect size is large enough it can still pull 
up the overall estimate.

In this application, we showed that the instrument has a naturally interpreted 
heterogeneous effect due to differences in baseline pollution levels. We showed 
that the data is consistent with sorting, and that comparing our CRC estimates 
with TSLS allows us to draw conclusions about effects for counties where the 
instrument is weak. Overall, our findings suggest that there is substantial 
heterogeneity in the size of the effect of pollution on housing prices.

%%%%%%%%%%%%%%%%%%%%%%%%%%%%%%%%%%%%%%%%%%%%%%%%%%%%%%%%%%%%%%%%%%%%%%%%%%%%%%%%
%%%%%%%%%%%%%%%%%%%%%%%%%%%%%%%%%%%%%%%%%%%%%%%%%%%%%%%%%%%%%%%%%%%%%%%%%%%%%%%%

\section{Conclusion}
\label{sec_conclusion}

In this paper we have studied a linear correlated random coefficients model. We 
provided conditions under which we can point identify the average partial and 
treatment effects of an endogenous treatment variable by using variation in an 
instrumental variable. In contrast to previous research, these conditions allow 
for heterogeneous effects of the instrument in the first stage equation, as well 
as binary or discrete instruments in many cases. Our identification argument led 
directly to a simple estimator of a trimmed average of the outcome coefficients. 
This estimator is just an average of weighted least squares regressions, where 
the weights depend on a first stage estimator. We established 
$\sqrt{n}$-consistency and asymptotic normality of this estimator, and showed 
that it performs well in finite sample simulations. We have illustrated how allowing for and analyzing heterogeneity in the first stage and in the outcome equation can be easily and fruitfully incorporated into a typical applied instrumental variables analysis. 

Several issues remain for future research. First, it may be theoretically 
interesting to modify our estimator to better account for boundary effects both 
due to kernel smoothing and first stage rank estimation. Second, we have not 
provided a method for choosing the bandwidth $h$, which is an important question 
in practice. Third, we assumed the set $\mathcal{R}$ for which the relevance 
condition \ref{as_rel} holds is known \textit{a priori}. In some applications, 
this is a reasonable assumption, such as in our empirical 
application. In other applications, it may not be reasonable. In principle, this 
set can be estimated in a preliminary step, and then the previous analysis can 
be repeated by using this estimated set $\widehat{\mathcal{R}}$ in place of 
$\mathcal{R}$. This extension is both nontrivial and of independent interest, 
and we leave it to future work. Fourth, it may be helpful to explore modifications of our proposed estimator to achieve efficiency gains. Finally, we are coding a Stata module that will enable practitioners to apply the estimator in this paper with minimal 
investment.

%%%%%%%%%%%%%%%%%%%%%%%%%%%%%%%%%%%%%%%%%%%%%%%%%%%%%%%%%%%%%%%%%%%%%%%%%%%%%%%%
%%%%%%%%%%%%%%%%%%%%%%%%%%%%%%%%%%%%%%%%%%%%%%%%%%%%%%%%%%%%%%%%%%%%%%%%%%%%%%%%

\appendix

\section{Proofs}
\label{ap_proofs}

\begin{proof}[\textbf{Proof of Proposition \ref{prop_controlvariable}}]
    For any $r \equiv (r_{1},\ldots,r_{d_{b}}) \in \interior \supp (R)$, $z \in 
    \supp(Z)$ and $b \in \re^{d_{w}}$ we have
    \begin{align*}
        \Prob\left[R_{k} \leq r_{k} \, \forall k, B \leq b \vert Z = z\right] &= 
        \Prob\left[X_{k} \leq Q_{X_{k} \vert Z}(r_{k} \vert z) \, 
            k=1,\ldots,d_{b}, B \leq b \vert Z = z\right] \\
        &= \Prob\left[h_{k}(z,V_{k}) \leq h_{k}(z,Q_{V_{k} \vert Z}(r_{k} \vert 
            z))\, \forall k, B \leq b \vert Z = z\right] \\
        &= \Prob\left[V_{k} \leq Q_{V_{k} \vert Z}(r_{k} \vert z)\, \forall k, B 
            \leq b \vert Z = z\right] \\
        &= \Prob\left[V_{k} \leq Q_{V_{k}}(r_{k})\, \forall k, B \leq b\right],
    \end{align*}
    where the first equality follows because for $k = 1,\ldots,d_{b}$, $X_{k} 
    \vert Z = z$ is continuous by \ref{as_fs} and $r_{k} \in (0,1)$, the second 
    follows by \ref{as_fs} and the equivariance to monotone transformations 
    property of quantiles, the third follows because $h_{k}(z, \cdot)$ is 
    strictly increasing by \ref{as_fs}, and the fourth uses \ref{as_ex}. Since 
    the right-hand side does not depend on $z$, we conclude that $(R, B) 
    \independent Z$. 
    
    The second statement follows because $R = r$ if and only if $X_{k} = 
    Q_{X_{k} \vert Z}(r_{k} \vert Z)$ for $k = 1,\ldots,d_{b}$. Hence, 
    conditional on $R = r$, $X_{k}$ is a stochastic function of $Z$ for $k = 
    1,\ldots,d_{b}$. Since, by \ref{as_derived}, the derived endogenous variables, $X_{k}$, $k > 
    d_{b}$, are functions of $X_{1},\ldots,X_{d_{b}}$ and $Z_{1}$, they are also 
    stochastic functions of $Z$ alone after conditioning on $R = r$. Thus, 
    conditional on $R = r$, $W \equiv [1, X', Z_{1}']'$ is only stochastic 
    through $Z$. Since $Z \independent B \vert R$ (as implied by $(R,B) 
    \independent Z$), this shows that $W \independent B \vert R$ as well.
\end{proof}

%%%%%%%%%%%%%%%%%%%%%%%%%%%%%%%%%%%%%%%%%%%%%%%%%%%%%%%%%%%%%%%%%%%%%%%%%%%%%%%%
%%%%%%%%%%%%%%%%%%%%%%%%%%%%%%%%%%%%%%%%%%%%%%%%%%%%%%%%%%%%%%%%%%%%%%%%%%%%%%%%
%%%%%%%%%%%%%%%%%%%%%%%%%%%%%%%%%%%%%%%%%%%%%%%%%%%%%%%%%%%%%%%%%%%%%%%%%%%%%%%%

\begin{proof}[\textbf{Proof of Theorem \ref{thm_an}}]
    We begin by rewriting the model as $Y_{i} \equiv W_{i}'\beta(r) + U_{i}(r)$, 
    where $U_{i}(r) \equiv W_{i}'(B_{i} - \beta(r))$. Substituting into 
    \eqref{eq_betahatr}, we have $\widehat{\beta}(r) = 
    \widehat{P}(r)^{+}\widehat{P}(r)\beta(r) + \widehat{P}(r)^{+}\widehat{A}(r)$ 
    with $\widehat{P}(r) \equiv 
    n^{-1}\sum_{i=1}^{n}\widehat{k}_{i}^{h}(r)W_{i}W_{i}'$ and $\widehat{A}(r) 
    \equiv n^{-1}\sum_{i=1}^{n}\widehat{k}_{i}^{h}(r)W_{i}U_{i}(r)$. For any 
    square matrix $M$, let $\sigma(M)$ denote the absolute value of the smallest 
    eigenvalue of $M$.  Then \ref{as_rel} and \ref{as_smoothness} imply that 
    $\inf_{r \in \mathcal{R}} \sigma(P(r)) = C$ for some $C > 0$. Defining
    $\widehat{1} \equiv \IndicSmall{\inf_{r \in \mathcal{R}} 
        \sigma(\widehat{P}(r)) > C/2}$, we have $\widehat{1}(\widehat{\beta}(r) 
    - \beta(r)) = \widehat{1}\widehat{P}(r)^{+}\widehat{A}(r)$,
    since $\widehat{1}(\widehat{P}(r)^{+}\widehat{P}(r) - I) = 0$ for all $r \in 
    \mathcal{R}$.\footnote{That $\widehat{1}$ is indeed a random variable (in 
        the sense of being a measurable function on the underlying sample space) 
        follows from Theorem 18.19 of \cite{aliprantisborder2006} combined with 
        standard results.}
    Centering, integrating over $\mathcal{R}$ and scaling by $\sqrt{n}$, we 
    obtain by Lemma \ref{lem_negligible} that
    \begin{align*}
        \widehat{1}\sqrt{n}\left(\int_{\R}\widehat{\beta}(r)\, dr - \int_{\R} 
            \beta(r)\,dr\right) = 
        \widehat{1}\sqrt{n}\int_{\R}P(r)^{-1}\widehat{A}(r)\, dr + \op{1}.
    \end{align*}

    Next, define $T_{i}(r) \equiv P(r)^{-1}W_{i}U_{i}(r)$ and note that our 
    assumptions only ensure that $T_{i}(r)$ is defined for $r \in \R$. For $r 
    \notin \R$ we use the convention that $\IndicSmall{r \in \R}T_{i}(r) = 0 
    \cdot T_{i}(r) = 0$, which will help to ease notation.\footnote{The 
        alternative would be to define $T_{i}(r)$ and related functions (its 
        derivatives, etc.) case by case, which seems unnecessarily tedious.} 
    Then
    \begin{align}
        \label{eq_pa_equal_alpha}
        \int_{\R}P(r)^{-1}\sqrt{n}\widehat{A}(r)\, dr &= 
        \frac{1}{\sqrt{n}}\sum_{i=1}^{n}\int \frac{1}{h}K\left(\frac{r - 
                \widehat{R}_{i}}{h}\right)\IndicSmall{r \in \R} T_{i}(r)\, dr 
        \notag \\
        &= \frac{1}{\sqrt{n}}\sum_{i=1}^{n} \int K(\eta) 
        \IndicSmall{\widehat{R}_{i} + h\eta \in \R}T_{i}(\widehat{R}_{i} + 
        h\eta)\, d\eta \notag \\
        &\stackrel{\text{a.s.}}{=} \frac{1}{\sqrt{n}}\sum_{i=1}^{n} 
        \IndicSmall{R_{i} \notin \boundary(\R)} \int 
        K(\eta)\IndicSmall{\widehat{R}_{i} + h\eta \in \R}T_{i} (\widehat{R}_{i} 
        + h\eta)\, d\eta,
    \end{align}
    where the second equality follows by changing the variable of integration 
    from $r$ to $\eta \equiv (r - \widehat{R}_{i})/h$ and the third equality 
    follows with $\boundary(\R)$ denoting the boundary of $\R$ and 
    $\stackrel{\text{a.s.}}{=}$ denoting almost-sure equality because 
    $\Prob[R_{i} \in \boundary(\R), \text{ any } i] = 0$.\footnote{Since 
        almost-sure equality is sufficient for determining limiting 
        distributions, we will drop the ``a.s.'' qualifier in the following.}  
    Note that every component of the vector $\IndicSmall{r \in \R}T_{i}(r)$ is 
    twice continuously differentiable at all $r \notin \boundary(\R)$, since 
    \ref{as_rel} and \ref{as_smoothness} imply that $T_{i}(r)$ is twice 
    continuously differentiable at all $r \in \R^{\circ}$ (the interior of $\R$) 
    with first and second component-wise derivatives $\dot{T}_{i}(r)$ and 
    $\ddot{T}_{i}(r)$, while $\IndicSmall{r \in \R}T_{i}(r) = 0$ for $r \notin 
    \R$.\footnote{The differentiability of $T_{i}(r)$ for $r \in \R^{\circ}$ can 
        be determined using the calculus rules for matrices of functions derived 
        in Section 6.5 of \cite{hornjohnson1991cupc}.} For $R_{i} \in 
    \R^{\circ}$, a second order element-by-element application of Young's form 
    of Taylor's Theorem yields
    \begin{align}
        \label{eq_taylorexpansion}
        &\IndicSmall{\widehat{R}_{i} + h\eta \in \R}T_{i}(\widehat{R}_{i} + 
        h\eta) \notag \\
        &\quad = \IndicSmall{R_{i} \in \R}\Big[T_{i}(R_{i}) + (\widehat{R}_{i} - 
        R_{i}) \dot{T}_{i}(R_{i}) + h\eta \dot{T}_{i}(R_{i}) + (\widehat{R}_{i} 
        - R_{i})^{2}(\ddot{T}_{i}(R_{i}) + o(1)) \notag \\
        &\qquad\qquad\qquad \quad + 2(\widehat{R}_{i} - R_{i})h\eta 
        (\ddot{T}_{i}(R_{i}) + o(1)) + (h\eta)^{2}(\ddot{T}_{i}(R_{i}) + 
        o(1))\Big].
    \end{align}

    Let $\alpha_{i}^{h} \equiv \IndicSmall{R_{i} \notin \boundary(\R)}\int 
    K(\eta)\IndicSmall{\widehat{R}_{i} + h\eta \in \R}T_{i}(\widehat{R}_{i} + 
    h\eta)\, d\eta$ so that from \eqref{eq_pa_equal_alpha} we have 
    $\int_{\R}P(r)^{-1}\sqrt{n}\widehat{A}(r)\, dr = 
    n^{-1/2}\sum_{i=1}^{n}\alpha_{i}^{h}$. Substituting the Taylor expansion in 
    \eqref{eq_taylorexpansion} into the definition of $\alpha_{i}^{h}$ and using 
    the symmetry of $K$, $\int K(\eta)\, d\eta = 1$ and $\IndicSmall{R_{i} \in 
        \R}\IndicSmall{R_{i} \notin \boundary(\R)} = \IndicSmall{R_{i} \in 
        \R^{\circ}}$, we obtain
    \begin{align}
        \label{eq_aftertaylorexpansion}
        \frac{1}{\sqrt{n}}\sum_{i=1}^{n}\alpha_{i}^{h} &= 
        \frac{1}{\sqrt{n}}\sum_{i=1}^{n} \IndicSmall{R_{i} \in \R^{\circ}} 
        \left[T_{i}(R_{i}) + (\widehat{R}_{i} - R_{i})\dot{T}_{i}(R_{i})\right] 
        \\
        &\quad + \frac{1}{\sqrt{n}}\sum_{i=1}^{n} \IndicSmall{R_{i} \in 
            \R^{\circ}}\left[(\widehat{R}_{i} - R_{i})^{2} + h^{2}\int 
            \eta^{2}K(\eta)\, d\eta\right](\ddot{T}_{i}(R_{i}) + o(1)). \notag
    \end{align}
    It can be shown through some tedious algebra that \ref{as_rel}, 
    \ref{as_smoothness} and \ref{as_moments_est} imply that 
    $\Exp[\IndicSmall{R_{i} \in 
        \mathcal{R}^{\circ}}\normSmall{\ddot{T}_{i}(R_{i})}]$ is finite. As a 
    result, the second term in \eqref{eq_aftertaylorexpansion} is $\op{1}$, 
    since by \ref{as_rank_estimation}
    \begin{align*}
        &\norm{\frac{1}{\sqrt{n}}\sum_{i=1}^{n}\IndicSmall{R_{i} \in 
                \R^{\circ}}(\widehat{R}_{i} - R_{i})^{2}(\ddot{T}_{i}(R_{i}) + 
            o(1))} \\
        &\qquad \leq \max_{i : R_{i} \in \R^{\circ}} 
        \sqrt{n}(\widehat{R}_{i} - 
            R_{i})^{2}\left[\frac{1}{n}\sum_{i=1}^{n}\IndicSmall{R_{i} \in 
                \R^{\circ}}\left(\normSmall{\ddot{T}_{i}(R_{i})} + 
                o(1)\right)\right] = \Op{n^{-1/2}},
    \end{align*}
    and by $\int \eta^{2}K(\eta)\,d\eta$ finite and \ref{as_bandwidth},
    \begin{align*}
        &\norm{\frac{1}{\sqrt{n}}\sum_{i=1}^{n}\IndicSmall{R_{i} \in 
                \R^{\circ}}\left[h^{2}\int\eta^{2}K(\eta)\, 
                d\eta\right](\ddot{T}_{i}(R_{i}) + o(1))} \\
        &\qquad \leq 
        \sqrt{n}O(h^{2})\left[\frac{1}{n}\sum_{i=1}^{n}\IndicSmall{R_{i} \in 
                \R^{\circ}}\left(\normSmall{\ddot{T}_{i}(R_{i})} + 
                o(1)\right)\right] = \op{1}.
    \end{align*}
    
    Substituting the asymptotically linear form for $\widehat{R}_{i} - R_{i}$ 
    given in \ref{as_rank_estimation} into the first term in 
    \eqref{eq_aftertaylorexpansion}, we obtain
    \begin{align}
        \label{eq_an_vstat_plus_neg}
        \frac{1}{\sqrt{n}}\sum_{i=1}^{n}\alpha_{i}^{h} &= 
        \frac{\sqrt{n}}{n^{2}}\sum_{i=1}^{n}\sum_{j=1}^{n}\IndicSmall{R_{i} \in 
            \mathcal{R}^{\circ}}\left( T_{i}(R_{i}) + \xi_{j}(X_{i} \vert 
            Z_{i})\dot{T}_{i}(R_{i})\right)  \notag \\
        &\qquad + \frac{1}{n}\sum_{i=1}^{n}
        \IndicSmall{R_{i} \in \mathcal{R}^{\circ}} \rho_{n}(X_{i} \vert 
        Z_{i})\dot{T}_{i}(R_{i})  + \op{1}.
    \end{align}
    Some tedious algebra shows that \ref{as_rel}, \ref{as_smoothness} and 
    \ref{as_moments_est} imply that $\Exp[\IndicSmall{R_{i} \in 
        \mathcal{R}^{\circ}}\normSmall{\dot{T}_{i}(R_{i})}^{4}]$ is finite.
    Consequently, the second term in \eqref{eq_an_vstat_plus_neg} is 
    asymptotically negligible under \ref{as_rank_estimation} since
    \begin{align*}
        &\norm{\frac{1}{n}\sum_{i=1}^{n}
            \IndicSmall{R_{i} \in \mathcal{R}^{\circ}} \rho_{n}(X_{i} \vert 
            Z_{i})\dot{T}_{i}(R_{i})} \\
        &\quad \leq  \left[\sup_{(x,z) \in \mathcal{XZ}(\R)}\absSmall{\rho_{n}(x 
                \vert z)}\right]
        \frac{1}{n}\sum_{i=1}^{n} \IndicSmall{R_{i} \in \mathcal{R}^{\circ}} 
        \normSmall{\dot{T}_{i}(R_{i})} = \op{1}.
    \end{align*}
    The first term in \eqref{eq_an_vstat_plus_neg} can be written as a 
    second-order V-statistic with a symmetric kernel by defining $M_{ij} \equiv 
    \IndicSmall{R_{i} \in \mathcal{R}^{\circ}}\left( T_{i}(R_{i}) + 
        \xi_{j}(X_{i} \vert Z_{i})\dot{T}_{i}(R_{i})\right)$ and $\zeta_{ij} 
    \equiv \frac{1}{2}(M_{ij} + M_{ji})$, so that
    \begin{align*}
        \frac{\sqrt{n}}{n^{2}}\sum_{i=1}^{n}\sum_{j=1}^{n}\IndicSmall{R_{i} \in 
            \mathcal{R}^{\circ}}\left( T_{i}(R_{i}) + \xi_{j}(X_{i} \vert 
            Z_{i})\dot{T}_{i}(R_{i})\right) =  
        \frac{\sqrt{n}}{n^{2}}\sum_{i=1}^{n}\sum_{j=1}^{n} \zeta_{ij}.
    \end{align*}
    As noted, it can be shown that $\Exp[\IndicSmall{R_{i} \in 
        \mathcal{R}^{\circ}} \normSmall{\dot{T}_{i}(R_{i})}^{4}]$ is finite 
    under our moment and smoothness assumptions. It can similarly be shown that 
    \ref{as_rel}, \ref{as_smoothness} and \ref{as_moments_est} imply that 
    $\Exp[\IndicSmall{R_{i} \in 
        \mathcal{R}^{\circ}}\normSmall{T_{i}(R_{i})}^{2}]$ is also finite. These 
    observations and \ref{as_rank_estimation} imply that
    \begin{align*}
        \Exp(\normSmall{M_{ij}}^{2}) &\leq \Exp\left(\IndicSmall{R_{i} \in 
                \R^{\circ}}\normSmall{T_{i}(R_{i})}^{2}\right) \\
        &\qquad + \Exp\left(\IndicSmall{R_{i} \in 
                \R^{\circ}}\normSmall{\dot{T}_{i}(R_{i})}^{4}\right)^{1/2}
        \Exp\left(\xi_{j}(X_{i} \vert Z_{i})^{4}\right)^{1/2} < \infty,
    \end{align*}
    for both $j \neq i$ and $j = i$. Since $\Exp(\normSmall{\zeta_{ij}}^{2}) 
    \leq \frac{1}{2}\Exp(\normSmall{M_{ij}}^{2})$, this implies that 
    $\Exp(\normSmall{\zeta_{ij}}^{2}) < \infty$ for all $i$ and $j$, by which we 
    can conclude that
    \begin{align*}
        \frac{\sqrt{n}}{n^{2}}\sum_{i=1}^{n}\sum_{j=1}^{n} \zeta_{ij} = 
        \sqrt{n}{n \choose 2}^{-1}\sum_{i < j} \zeta_{ij} + \op{1}.
    \end{align*}
    That is, the second-order V-statistic is asymptotically equivalent to a 
    second-order U-statistic, see e.g.\ pg. 206 of 
    \cite{serfling1980}.\footnote{Serfling's discussion is limited to 
        scalar-valued variables, but the modification for vector-valued 
        variables is immediate.} Recapping, we have now shown that
    \begin{align}
        \label{eq_ustat}
        \widehat{1}\sqrt{n}\left(\int_{\R}\widehat{\beta}(r)\, dr - \int_{\R} 
            \beta(r)\,dr\right) = \widehat{1}\sqrt{n}{n \choose 2}^{-1}\sum_{i < 
            j} \zeta_{ij} + \op{1}.
    \end{align}
    
    We determine the projection of the U-statistic in \eqref{eq_ustat} by 
    computing $\Exp(\zeta_{ij} \vert i)$ for $j \neq i$, where the notation 
    $\Exp(\cdot \vert i)$ is shorthand for $\Exp(\cdot \vert Y_{i}, X_{i}, 
    Z_{i})$. First, since $\Exp[\xi_{j}(X_{i} \vert Z_{i}) \vert i] = 0$ for 
    $R_{i} \in \R^{\circ}$ by \ref{as_iid} and \ref{as_rank_estimation},
    \begin{align*}
        \Exp(M_{ij} \vert i) &= \Exp\left[\IndicSmall{R_{i} \in 
                \mathcal{R}^{\circ}}\left( T_{i}(R_{i}) + \xi_{j}(X_{i} \vert 
                Z_{i})\dot{T}_{i}(R_{i})\right) \vert i \right] \\
        &= \IndicSmall{R_{i} \in \mathcal{R}^{\circ}}\left(T_{i}(R_{i}) + 
            \Exp(\xi_{j}(X_{i} \vert Z_{i}) \vert i)\dot{T}_{i}(R_{i})\right) 
        \stackrel{\text{a.s.}}{=} \IndicSmall{R_{i} \in \mathcal{R}}T_{i}(R_{i}) 
        \equiv \zeta_{1i},
    \end{align*}
    which has mean zero, because $\Exp(\zeta_{1i} \vert R_{i} = r) = 
    \IndicSmall{r \in \mathcal{R}}P(r)\Exp(W_{i}U_{i}(r) \vert R_{i} = r) = 0$ 
    and $\Exp(W_{i}U_{i}(r) \vert R_{i} = r) = \Exp(W_{i}W_{i}'(B_{i} - 
    \beta(r)) \vert R_{i} = r) = 0$ by Proposition \ref{prop_controlvariable}.  
    Second,
    \begin{align*}
        \Exp(M_{ji} \vert i) &= \Exp\left[\IndicSmall{R_{j} \in 
                \mathcal{R}^{\circ}}\left( T_{j}(R_{j}) + \xi_{i}(X_{j} \vert 
                Z_{j})\dot{T}_{j}(R_{j})\right) \vert i \right] \\
        &= \Exp[\IndicSmall{R_{j} \in \mathcal{R}^{\circ}}T_{j}(R_{j})] + 
        \Exp\left[\IndicSmall{R_{j} \in \mathcal{R}^{\circ}}\xi_{i}(X_{j} \vert 
            Z_{j})\dot{T}_{j}(R_{j}) \vert i \right] \\
        &= \Exp(\zeta_{1j}) + \Exp\left[\IndicSmall{R_{j} \in \R^{\circ}} 
            \xi_{i}(X_{j} \vert Z_{j})\Exp(\dot{T}_{j}(R_{j}) \vert R_{j}, 
            Z_{j}) \vert i \right] \\
        &= -\Exp[\IndicSmall{R_{j} \in \R}\xi_{i}(X_{j} \vert 
        Z_{j})P(R_{j})^{-1}W_{j}W_{j}'\dot{\beta}(R_{j}) \vert i] \equiv 
        \zeta_{2i},
    \end{align*}
    where the third equality uses the law of iterated expectations (noting that 
    $X_{j}$ is deterministic conditional on $R_{j}, Z_{j}$) and the fourth 
    equality uses $\dot{T}_{j}(r) = -P(r)^{-1}[\dot{P}(r)P(r)^{-1}W_{j}U_{j}(r) 
    + W_{j}W_{j}'\dot{\beta}(r)]$ and $\Exp[U_{j}(R_{j}) \vert R_{j}, Z_{j}] = 
    0$.\footnote{This expression for the component-wise derivative of $T_{j}(r)$ 
        follows immediately from the rules in Section 6.5 of 
        \cite{hornjohnson1991cupc}.} Applying the law of iterated expectations 
    shows that $\Exp(\zeta_{2i}) = 0$, since $\Exp(\xi_{i}(x \vert z)) = 0$ for 
    $(x,z) \in \mathcal{XZ}(\R)$ by \ref{as_rank_estimation}. We have now shown 
    that $\Exp(\zeta_{ij} \vert i) = \frac{1}{2}(\zeta_{1i} + \zeta_{2i})$, 
    $\Exp(\zeta_{ij}) = \Exp[\Exp(\zeta_{ij} \vert i)] = 0$ and 
    $\Exp(\normSmall{\zeta_{ij}}^{2}) < \infty$, so that by the central limit 
    theorem for U-statistics (e.g.\ Theorem A on page 192 of 
    \citealt{serfling1980})\footnote{Serfling's discussion is limited to 
        scalar-valued variables. The original work by \cite{hoeffding1948taoms} 
        contains an explicit statement of the vector case; see also Chapter 5 of 
        \cite{kowalskitu2008} for a modern treatment.}
    \begin{align}
        \label{eq_ustat_limiting_dist}
        \sqrt{n}{n \choose 2}^{-1}\sum_{i < j} \zeta_{ij} \rightsquigarrow 
        N\left(0, \Exp[(\zeta_{1i} + \zeta_{2i})(\zeta_{1i} + 
            \zeta_{2i})']\right).
    \end{align}
    Combining equations \eqref{eq_ustat} and \eqref{eq_ustat_limiting_dist} and 
    applying Slutsky's theorem with $\widehat{1} \plim 1$ from Lemma 
    \ref{lem_negligible}, we have
    \begin{align*}
        \sqrt{n}\left(\int_{\R} \widehat{\beta}(r)\, dr - \int_{\R} \beta(r)\, 
            dr\right) \rightsquigarrow N\left(0, \Exp[(\zeta_{1i} + 
            \zeta_{2i})(\zeta_{1i} + \zeta_{2i})']\right).
    \end{align*}
    The result now follows after scaling both sides by $\lambda(\R)^{-1}$.
\end{proof}

%%%%%%%%%%%%%%%%%%%%%%%%%%%%%%%%%%%%%%%%%%%%%%%%%%%%%%%%%%%%%%%%%%%%%%%%%%%%%%%%
%%%%%%%%%%%%%%%%%%%%%%%%%%%%%%%%%%%%%%%%%%%%%%%%%%%%%%%%%%%%%%%%%%%%%%%%%%%%%%%%

\begin{lemma}
    \label{lem_negligible}
    Under the assumptions of Theorem \ref{thm_an}, $\widehat{1} \plim 1$ and
    \begin{align*}
        \widehat{1}\int_{\R}[\widehat{P}(r)^{+} - 
        P(r)^{-1}]\widehat{A}(r)\, dr = \op{n^{-1/2}}.
    \end{align*}
\end{lemma}

\begin{proof}[\textbf{Proof of Lemma \ref{lem_negligible}}]
    Let $J(r) \equiv \widehat{P}(r)^{+} - P(r)^{-1}$, let $\norm{\cdot}_{1}$ 
    denote the $l_{1}$ norm and let $\norm{\cdot}_{1,op}$ denote the matrix 
    operator norm induced by the $l_{1}$ norm. Then
    \begin{align}
        \label{eq_maindecomposition}
        \norm{\widehat{1}\int_{\R}J(r)\widehat{A}(r)\, dr}_{1} &\leq 
        \widehat{1}\int_{\R}\normSmall{J(r)\widehat{A}(r)}_{1}\, dr \notag \\
        &\leq 
        \widehat{1}\int_{\R}\normSmall{J(r)}_{op,1}\normSmall{\widehat{A}(r)}_{1}\,dr 
        \leq \widehat{1}\sup_{r \in \R} \normSmall{J(r)}_{op,1}\sup_{r \in \R} 
        \normSmall{\widehat{A}(r)},
    \end{align}
    where the first inequality follows because for vector-valued function $x: 
    \re \rightarrow \re^{K}$ and component-wise integration we have 
    $\sum_{k=1}^{K} \absSmall{\int x(r)\, dr} \leq \sum_{k=1}^{K} \int 
    \absSmall{x(r)}\, dr = \int\normSmall{x(r)}_{1}\, dr$, and the second 
    inequality uses the sub-multiplicative property of the matrix operator norm.  

    First, we consider the behavior of
    \begin{align*}
        \sup_{r \in \R}\normSmall{\widehat{P}(r) - P(r)} \leq \sup_{r \in 
            \R}\normSmall{\widehat{P}(r) - \widetilde{P}(r)} + \sup_{r \in 
            \R}\normSmall{\widetilde{P}(r) - P(r)},
    \end{align*}
    where $\widetilde{P}(r) \equiv n^{-1}\sum_{i=1}^{n}k_{i}^{h}(r)W_{i}W_{i}'$ 
    with $k_{i}^{h} \equiv h^{-1}K((R_{i} - r)/h)$. Notice that both 
    $\widehat{P}(r)$ and $\widetilde{P}(r)$ are Nadaraya-Watson kernel 
    regression estimators of the matrix $P(r) \equiv \Exp[WW' \vert R = 
    r]f_{R}(r)$ with $f_{R}(r) = 1$, but that the weights in $\widehat{P}(r)$ 
    use the generated regressor $\widehat{R}_{i}$, while the weights in 
    $\widetilde{P}(r)$ use $R_{i}$.  Recent work on nonparametric regression 
    with generated regressors has established that $\sup_{r \in 
        \R}\normSmall{\widehat{P}(r) - \widetilde{P}(r)} = \Op{\log(n)n^{-1/2}}$ 
    under our assumptions.\footnote{\label{fn_mammen} In particular, we appeal 
        to Lemma 1 of \cite{mammenrotheschienle2012taos}, but see also 
        \cite{sperlich2009ej}, \cite{mammenrotheschienle2013wp}, 
        \cite{hahnridder2013e}, \cite{lee2013wp} and 
        \cite{escancianojacho-chavezlewbel2014joe} for related results. We 
        verify the conditions for Lemma 1 of \cite{mammenrotheschienle2012taos}.  
        In their Assumption 1, (i) is \ref{as_iid}, (ii) is satisfied with $R 
        \sim \text{Unif}[0,1]$, (iii) is \ref{as_smoothness}, (iv) is not used 
        in the proof of their Lemma 1, (v) is \ref{as_kernel} and (vi) is met 
        under \ref{as_bandwidth} and \ref{as_rank_estimation}. Their Assumptions 
        2 and 3 are satisfied by our \ref{as_rank_estimation}, while their 
        Assumption 4 is not used in the proof of their Lemma 1. The rate of 
        $\Op{\log(n)n^{-1/2}}$ is determined by computing $\kappa_{1}$ on pg.  
        1141 and observing their notational convention of leaving out $\log(n)$ 
        terms.} 
    Using standard results in the literature, our assumptions also ensure that 
    $\sup_{r \in \R} \normSmall{\widetilde{P}(r) - P(r)} = 
    \Op{(\log(n)/nh)^{1/2} + h^{2}}$, with the dominant rate being 
    $(\log(n)/nh)^{1/2}$ given \ref{as_bandwidth}.\footnote{For example, see 
        Lemma B3 of \cite{newey1994et}.} Since this rate also dominates 
    $\log(n)n^{-1/2}$, it follows that $\sup_{r \in \mathcal{R}} 
    \absSmall{\widehat{P}(r) - P(r)} = \Op{(\log(n)/nh)^{1/2}}$. Note that this 
    also implies that $\widehat{1} \plim 1$. 
    Using these observations, the definition of $\widehat{1}$ and \ref{as_rel} 
    with \ref{as_smoothness}, we have that
    \begin{align*}
        \widehat{1}\sup_{r \in \R}\normSmall{J(r)} &= \widehat{1}\sup_{r \in 
            \R}\normSmall{\widehat{P}(r)^{-1}(P(r) - \widehat{P}(r))P(r)^{-1}} 
        \\
        &\leq \widehat{1}\sup_{r \in \R}\normSmall{\widehat{P}(r)^{-1}}\sup_{r 
            \in \R}\normSmall{\widehat{P}(r) - P(r)}\sup_{r \in 
            \R}\normSmall{P(r)^{-1}} = \Op{(\log(n)/nh)^{1/2}}.
    \end{align*}
    The same rate applies to $\widehat{1}\sup_{r \in \R} 
    \normSmall{J(r)}_{op,1}$, because finite-dimensional norms are equivalent.

    A rate of convergence for $\sup_{r \in \R} \normSmall{\widehat{A}(r)}$ 
    follows similarly after using the definition of $U_{i}(r) \equiv 
    W_{i}'(B_{i} - \beta(r))$ to write
    \begin{align*}
        \widehat{A}(r) &= 
        \frac{1}{n}\sum_{i=1}^{n}\widehat{k}_{i}^{h}(r)W_{i}W_{i}'B_{i} - 
        \widehat{P}(r)\beta(r) \\
        \text{ and } \quad \widetilde{A}(r) &= 
        \frac{1}{n}\sum_{i=1}^{n}k_{i}^{h}(r)W_{i}W_{i}'B_{i} - 
        \widetilde{P}(r)\beta(r).
    \end{align*}
    The difference of the first two terms in these expressions is uniformly 
    $\Op{\log(n)n^{-1/2}}$ again by Lemma 1 of 
    \cite{mammenrotheschienle2012taos}.\footnote{The verification for most of 
        their conditions is as in footnote \ref{fn_mammen}. Their
        Assumption 1 (iii) is satisfied since each component of $\Exp[WW'B \vert 
        R = r] = P(r)\beta(r)$ is twice continuously differentiable by 
        \ref{as_smoothness}.} Since $\beta(r)$ is bounded uniformly over $\R$ by 
    \ref{as_rset} and \ref{as_smoothness}, the difference of the second two 
    terms is also $\Op{\log(n)n^{-1/2}}$ using the already established rate for 
    $\sup_{r \in \R}\normSmall{\widehat{P}(r) - \widetilde{P}(r)}$, and hence 
    $\sup_{r \in \R}\normSmall{\widehat{A}(r) - \widetilde{A}(r)} = 
    \Op{\log(n)n^{-1/2}}$. Also, since $\Exp[WW'B \vert R = r] - \Exp[WW' \vert 
    R = r]\beta(r) = 0$ by Proposition \ref{prop_controlvariable}, standard 
    results again imply that $\sup_{r \in \R} \normSmall{\widetilde{A}(r)} = 
    \Op{(\log(n)/nh)^{1/2}}$ under our assumptions. We conclude that $\sup_{r 
        \in \R} \normSmall{\widehat{A}(r)} = \Op{(\log(n)/nh)^{1/2}}$. This 
    establishes the claim via \eqref{eq_maindecomposition}, since under 
    \ref{as_bandwidth}, $\sqrt{n}\log(n)(nh)^{-1} = \log(n)/(\sqrt{n}h) 
    \rightarrow 0$.
\end{proof}

%%%%%%%%%%%%%%%%%%%%%%%%%%%%%%%%%%%%%%%%%%%%%%%%%%%%%%%%%%%%%%%%%%%%%%%%%%%%%%%%
%%%%%%%%%%%%%%%%%%%%%%%%%%%%%%%%%%%%%%%%%%%%%%%%%%%%%%%%%%%%%%%%%%%%%%%%%%%%%%%%

\singlespacing
\small

\nocite{2010JEPsymposium}
\nocite{2010JELsymposium}
\bibliographystyle{ecta}

\newpage
\bibliography{crc}

\begin{thebibliography}{61}
\newcommand{\enquote}[1]{``#1''}
\expandafter\ifx\csname natexlab\endcsname\relax\def\natexlab#1{#1}\fi

\bibitem[\protect\citeauthoryear{Aliprantis and Border}{Aliprantis and
  Border}{2006}]{aliprantisborder2006}
\textsc{Aliprantis, C.~D. and K.~C. Border} (2006): \emph{{I}nfinite
  {D}imensional {A}nalysis: {A} {H}itchhiker's {G}uide}, Springer, 3 ed.

\bibitem[\protect\citeauthoryear{Angrist and Imbens}{Angrist and
  Imbens}{1995}]{angristimbens1995jotasa}
\textsc{Angrist, J.~D. and G.~W. Imbens} (1995): \enquote{{T}wo-{S}tage {L}east
  {S}quares {E}stimation of {A}verage {C}ausal {E}ffects in {M}odels with
  {V}ariable {T}reatment {I}ntensity,} \emph{The Journal of the American
  Statistical Association}, 90, 431--442.

\bibitem[\protect\citeauthoryear{Angrist, Keane, Leamer, Nevo, Pischke, Sims,
  Stock, and Whinston}{Angrist et~al.}{2010}]{2010JEPsymposium}
\textsc{Angrist, J.~D., M.~P. Keane, E.~E. Leamer, A.~Nevo, J.-S. Pischke,
  C.~A. Sims, J.~H. Stock, and M.~D. Whinston} (2010): \enquote{{S}ymposia:
  {C}on out of {E}conomics,} \emph{The Journal of Economic Perspectives}, 24,
  3--94.

\bibitem[\protect\citeauthoryear{Beran and Hall}{Beran and
  Hall}{1992}]{beranhall1992taos}
\textsc{Beran, R. and P.~Hall} (1992): \enquote{{E}stimating {C}oefficient
  {D}istributions in {R}andom {C}oefficient {R}egressions,} \emph{The Annals of
  Statistics}, 1970--1984.

\bibitem[\protect\citeauthoryear{Blundell and Powell}{Blundell and
  Powell}{2004}]{blundellpowell2004troes}
\textsc{Blundell, R.~W. and J.~L. Powell} (2004): \enquote{{E}ndogeneity in
  {S}emiparametric {B}inary {R}esponse {M}odels,} \emph{The Review of Economic
  Studies}, 71, 655--679.

\bibitem[\protect\citeauthoryear{Chamberlain}{Chamberlain}{1992}]{chamberlain1992e}
\textsc{Chamberlain, G.} (1992): \enquote{{E}fficiency {B}ounds for
  {S}emiparametric {R}egression,} \emph{Econometrica}, 60, 567--596.

\bibitem[\protect\citeauthoryear{Chay and Greenstone}{Chay and
  Greenstone}{2005}]{chaygreenstone2005jope}
\textsc{Chay, K.~Y. and M.~Greenstone} (2005): \enquote{{D}oes {A}ir {Q}uality
  {M}atter? {E}vidence from the {H}ousing {M}arket,} \emph{The Journal of
  Political Economy}, 113, 376--424.

\bibitem[\protect\citeauthoryear{Chernozhukov, Fern{\'a}ndez-Val, and
  Galichon}{Chernozhukov et~al.}{2010}]{chernozhukovfernandez-valgalichon2010e}
\textsc{Chernozhukov, V., I.~Fern{\'a}ndez-Val, and A.~Galichon} (2010):
  \enquote{{Q}uantile and {P}robability {C}urves {W}ithout {C}rossing,}
  \emph{Econometrica}, 78, 1093--1125.

\bibitem[\protect\citeauthoryear{Chernozhukov, Fern{\'a}ndez-Val, and
  Kowalski}{Chernozhukov et~al.}{2011}]{chernozhukovfernandezvalkowalski2011wp}
\textsc{Chernozhukov, V., I.~Fern{\'a}ndez-Val, and A.~Kowalski} (2011):
  \enquote{{Q}uantile {R}egression with {C}ensoring and {E}ndogeneity,}
  \emph{Working paper}.

\bibitem[\protect\citeauthoryear{Chernozhukov, Fern{\'a}ndez-Val, and
  Melly}{Chernozhukov et~al.}{2009}]{chernozhukovfernandez-valmelly2009wp}
\textsc{Chernozhukov, V., I.~Fern{\'a}ndez-Val, and B.~Melly} (2009):
  \enquote{{I}nference on {C}ounterfactual {D}istributions,} \emph{Working
  paper}.

\bibitem[\protect\citeauthoryear{Chernozhukov, Fern{\'a}ndez-Val, and
  Melly}{Chernozhukov et~al.}{2012}]{chernozhukovfern'andez-valmelly2012e}
---\hspace{-.1pt}---\hspace{-.1pt}--- (2012): \enquote{{I}nference on
  {C}ounterfactual {D}istributions,} \emph{Econometrica (forthcoming)}.

\bibitem[\protect\citeauthoryear{Chernozhukov and Hansen}{Chernozhukov and
  Hansen}{2005}]{chernozhukovhansen2005e}
\textsc{Chernozhukov, V. and C.~Hansen} (2005): \enquote{{A}n {IV} {M}odel of
  {Q}uantile {T}reatment {E}ffects,} \emph{Econometrica}, 73, 245--261.

\bibitem[\protect\citeauthoryear{Chesher}{Chesher}{2003}]{chesher2003e}
\textsc{Chesher, A.} (2003): \enquote{{I}dentification in {N}onseparable
  {M}odels,} \emph{Econometrica}, 71, 1405--1441.

\bibitem[\protect\citeauthoryear{Deaton, Heckman, and Imbens}{Deaton
  et~al.}{2010}]{2010JELsymposium}
\textsc{Deaton, A., J.~J. Heckman, and G.~W. Imbens} (2010): \enquote{{F}orum
  on the {E}stimation of {T}reatment {E}ffects,} \emph{The Journal of Economic
  Literature}, 48, 356--455.

\bibitem[\protect\citeauthoryear{D'Haultf{\oe}uille and
  F{\'e}vrier}{D'Haultf{\oe}uille and
  F{\'e}vrier}{2012}]{d'haultfouillefevrier2012wp}
\textsc{D'Haultf{\oe}uille, X. and P.~F{\'e}vrier} (2012):
  \enquote{{I}dentification of {N}onseparable {M}odels with {E}ndogeneity and
  {D}iscrete {I}nstruments,} \emph{Working paper}.

\bibitem[\protect\citeauthoryear{Doksum}{Doksum}{1974}]{doksum1974taos}
\textsc{Doksum, K.} (1974): \enquote{{E}mpirical {P}robability {P}lots and
  {S}tatistical {I}nference for {N}onlinear {M}odels in the {T}wo-{S}ample
  {C}ase,} \emph{The Annals of Statistics}, 2, 267--277.

\bibitem[\protect\citeauthoryear{Ekeland, Heckman, and Nesheim}{Ekeland
  et~al.}{2004}]{EkelandHeckmanNesheim2004}
\textsc{Ekeland, I., J.~J. Heckman, and L.~Nesheim} (2004):
  \enquote{Identification and Estimation of Hedonic Models,} \emph{The Journal
  of Political Economy}, 112, 60.

\bibitem[\protect\citeauthoryear{Escanciano, Jacho-Chavez, and
  Lewbel}{Escanciano et~al.}{2014}]{escancianojacho-chavezlewbel2014joe}
\textsc{Escanciano, J.~C., D.~T. Jacho-Chavez, and A.~Lewbel} (2014):
  \enquote{{U}niform {C}onvergence of {W}eighted {S}ums of {N}on and
  {S}emiparametric {R}esiduals for {E}stimation and {T}esting,} \emph{Journal
  of Econometrics}, 178, 426--443.

\bibitem[\protect\citeauthoryear{Florens, Heckman, Meghir, and
  Vytlacil}{Florens et~al.}{2008}]{florensheckmanmeghiretal2008e}
\textsc{Florens, J.~P., J.~J. Heckman, C.~Meghir, and E.~Vytlacil} (2008):
  \enquote{{I}dentification of {T}reatment {E}ffects {U}sing {C}ontrol
  {F}unctions in {M}odels {W}ith {C}ontinuous, {E}ndogenous {T}reatment and
  {H}eterogeneous {E}ffects,} \emph{Econometrica}, 76, 1191--1206.

\bibitem[\protect\citeauthoryear{Graham and Powell}{Graham and
  Powell}{2012}]{grahampowell2012e}
\textsc{Graham, B.~S. and J.~L. Powell} (2012): \enquote{{I}dentification and
  {E}stimation of {A}verage {P}artial {E}ffects in “{I}rregular”
  {C}orrelated {R}andom {C}oefficient {P}anel {D}ata {M}odels,}
  \emph{Econometrica}, 80, 2105--2152.

\bibitem[\protect\citeauthoryear{Hahn and Ridder}{Hahn and
  Ridder}{2013}]{hahnridder2013e}
\textsc{Hahn, J. and G.~Ridder} (2013): \enquote{{A}symptotic {V}ariance of
  {S}emiparametric {E}stimators {W}ith {G}enerated {R}egressors,}
  \emph{Econometrica}, 81, 315--340.

\bibitem[\protect\citeauthoryear{Heckman and Vytlacil}{Heckman and
  Vytlacil}{1998}]{heckmanvytlacil1998tjohr}
\textsc{Heckman, J. and E.~Vytlacil} (1998): \enquote{{I}nstrumental
  {V}ariables {M}ethods for the {C}orrelated {R}andom {C}oefficient {M}odel:
  {E}stimating the {A}verage {R}ate of {R}eturn to {S}chooling {W}hen the
  {R}eturn is {C}orrelated with {S}chooling,} \emph{The Journal of Human
  Resources}, 33, 974--987.

\bibitem[\protect\citeauthoryear{Heckman}{Heckman}{2001}]{heckman2001tjope}
\textsc{Heckman, J.~J.} (2001): \enquote{{M}icro {D}ata, {H}eterogeneity, and
  the {E}valuation of {P}ublic {P}olicy: {N}obel {L}ecture,} \emph{The Journal
  of Political Economy}, 109, 673--748.

\bibitem[\protect\citeauthoryear{Heckman, Matzkin, and Nesheim}{Heckman
  et~al.}{2010}]{HeckmanMatzkinNesheim2010}
\textsc{Heckman, J.~J., R.~L. Matzkin, and L.~Nesheim} (2010):
  \enquote{Nonparametric Identification and Estimation of Nonadditive Hedonic
  Models,} \emph{Econometrica}, 78, 1569--1591.

\bibitem[\protect\citeauthoryear{Heckman, Smith, and Clements}{Heckman
  et~al.}{1997}]{heckmansmithclements1997troes}
\textsc{Heckman, J.~J., J.~Smith, and N.~Clements} (1997): \enquote{{M}aking
  the {M}ost {O}ut of {P}rogramme {E}valuations and {S}ocial {E}xperiments:
  {A}ccounting for {H}eterogeneity in {P}rogramme {I}mpacts,} \emph{The Review
  of Economic Studies}, 64, 487--535.

\bibitem[\protect\citeauthoryear{Hoderlein, Klemel{\"a}, and Mammen}{Hoderlein
  et~al.}{2010}]{hoderleinklemelamammen2010et}
\textsc{Hoderlein, S., J.~Klemel{\"a}, and E.~Mammen} (2010):
  \enquote{{A}nalyzing the {R}andom {C}oefficient {M}odel {N}onparametrically,}
  \emph{Econometric Theory}, 26, 804--837.

\bibitem[\protect\citeauthoryear{Hoderlein and Sherman}{Hoderlein and
  Sherman}{2013}]{hoderleinsherman2013cwp4}
\textsc{Hoderlein, S. and R.~Sherman} (2013): \enquote{{I}dentification and
  {E}stimation in a {C}orrelated {R}andom {C}oefficients Binary Response
  Model,} \emph{Working paper}.

\bibitem[\protect\citeauthoryear{Hoeffding}{Hoeffding}{1948}]{hoeffding1948taoms}
\textsc{Hoeffding, W.} (1948): \enquote{{A} {C}lass of {S}tatistics with
  {A}symptotically {N}ormal {D}istribution,} \emph{The Annals of Mathematical
  Statistics}, 19, 293--325.

\bibitem[\protect\citeauthoryear{Horn and Johnson}{Horn and
  Johnson}{1991}]{hornjohnson1991cupc}
\textsc{Horn, R.~A. and C.~R. Johnson} (1991): \emph{{T}opics in {M}atrix
  {A}nalysis}, Cambridge: Cambridge University Press.

\bibitem[\protect\citeauthoryear{Imbens}{Imbens}{2007}]{imbens2007ic}
\textsc{Imbens, G.} (2007): \enquote{{N}onadditive {M}odels with {E}ndogenous
  {R}egressors,} in \emph{Advances in Economics and Econometrics: Theory and
  Applications, Ninth World Congress}, ed. by R.~Blundell, W.~Newey, and
  T.~Persson, Cambridge University Press, New York, vol.~3.

\bibitem[\protect\citeauthoryear{Imbens and Angrist}{Imbens and
  Angrist}{1994}]{imbensangrist1994e}
\textsc{Imbens, G.~W. and J.~D. Angrist} (1994): \enquote{{I}dentification and
  {E}stimation of {L}ocal {A}verage {T}reatment {E}ffects,}
  \emph{Econometrica}, 62, 467--475.

\bibitem[\protect\citeauthoryear{Imbens and Newey}{Imbens and
  Newey}{2009}]{imbensnewey2009e}
\textsc{Imbens, G.~W. and W.~K. Newey} (2009): \enquote{{I}dentification and
  {E}stimation of {T}riangular {S}imultaneous {E}quations {M}odels {W}ithout
  {A}dditivity,} \emph{Econometrica}, 77, 1481--1512.

\bibitem[\protect\citeauthoryear{Jun}{Jun}{2009}]{jun2009joe}
\textsc{Jun, S.~J.} (2009): \enquote{{L}ocal {S}tructural {Q}uantile {E}ffects
  in a {M}odel with a {N}onseparable {C}ontrol {V}ariable,} \emph{Journal of
  Econometrics}, 151, 82--97.

\bibitem[\protect\citeauthoryear{Kasy}{Kasy}{2011}]{kasy2010et}
\textsc{Kasy, M.} (2011): \enquote{{I}dentification {I}n {T}riangular {S}ystems
  {U}sing {C}ontrol {F}unctions,} \emph{Econometric Theory}, 27, 663--671.

\bibitem[\protect\citeauthoryear{Kasy}{Kasy}{2013}]{kasy2013wp}
---\hspace{-.1pt}---\hspace{-.1pt}--- (2013): \enquote{{I}dentification in
  {G}eneral {T}riangular {S}ystems,} \emph{Working paper}.

\bibitem[\protect\citeauthoryear{Koenker and Bassett}{Koenker and
  Bassett}{1978}]{koenkerbassett1978e}
\textsc{Koenker, R. and G.~Bassett} (1978): \enquote{{R}egression {Q}uantiles,}
  \emph{Econometrica}, 46, 33--50.

\bibitem[\protect\citeauthoryear{Kowalski and Tu}{Kowalski and
  Tu}{2008}]{kowalskitu2008}
\textsc{Kowalski, J. and X.~M. Tu} (2008): \emph{{M}odern {A}pplied
  {U}-statistics}, Wiley.

\bibitem[\protect\citeauthoryear{Lee}{Lee}{2013}]{lee2013wp}
\textsc{Lee, Y.-Y.} (2013): \enquote{{P}artial {M}ean {P}rocesses with
  {G}enerated {R}egressors: {C}ontinuous {T}reatment {E}ffects and
  {N}onseparable {M}odels,} \emph{Working paper}.

\bibitem[\protect\citeauthoryear{Li and Tobias}{Li and
  Tobias}{2011}]{litobias2011joe}
\textsc{Li, M. and J.~L. Tobias} (2011): \enquote{{B}ayesian {I}nference in a
  {C}orrelated {R}andom {C}oefficients {M}odel: {M}odeling {C}ausal {E}ffect
  {H}eterogeneity with an {A}pplication to {H}eterogeneous {R}eturns to
  {S}chooling,} \emph{Journal of Econometrics}, 162, 345--361.

\bibitem[\protect\citeauthoryear{Mammen, Rothe, and Schienle}{Mammen
  et~al.}{2012}]{mammenrotheschienle2012taos}
\textsc{Mammen, E., C.~Rothe, and M.~Schienle} (2012): \enquote{{N}onparametric
  {R}egression {W}ith {N}onparametrically {G}enerated {C}ovariates,} \emph{The
  Annals of Statistics}, 40, 1132--1170.

\bibitem[\protect\citeauthoryear{Mammen, Rothe, and Schienle}{Mammen
  et~al.}{2013}]{mammenrotheschienle2013wp}
---\hspace{-.1pt}---\hspace{-.1pt}--- (2013): \enquote{{S}emiparametric
  {E}stimation with {G}enerated {C}ovariates,} \emph{Working paper}.

\bibitem[\protect\citeauthoryear{Masten}{Masten}{2012}]{masten2012jmp}
\textsc{Masten, M.~A.} (2012): \enquote{{R}andom {C}oefficients on {E}ndogenous
  {V}ariables in {S}imultaneous {E}quations {M}odels,} \emph{Working paper}.

\bibitem[\protect\citeauthoryear{Newey}{Newey}{1994}]{newey1994et}
\textsc{Newey, W.~K.} (1994): \enquote{{K}ernel {E}stimation of {P}artial
  {M}eans and a {G}eneral {V}ariance {E}stimator,} \emph{Econometric Theory},
  10, 233--253.

\bibitem[\protect\citeauthoryear{Palmquist}{Palmquist}{2005}]{Palmquist2005}
\textsc{Palmquist, R.~B.} (2005): \enquote{Property Value Models,}
  \emph{Handbook of Environmental Economics}, 2, 763--819.

\bibitem[\protect\citeauthoryear{Powell, Stock, and Stoker}{Powell
  et~al.}{1989}]{powellstockstoker1989e}
\textsc{Powell, J.~L., J.~H. Stock, and T.~M. Stoker} (1989):
  \enquote{{S}emiparametric {E}stimation of {I}ndex {C}oefficients,}
  \emph{Econometrica}, 57, 1403--1430.

\bibitem[\protect\citeauthoryear{Rosen}{Rosen}{1974}]{Rosen1974}
\textsc{Rosen, S.} (1974): \enquote{Hedonic Prices and Implicit Markets:
  Product Differentiation in Pure Competition,} \emph{The Journal of Political
  Economy}, 82, 34--55.

\bibitem[\protect\citeauthoryear{Rothe}{Rothe}{2009}]{rothe2009joe}
\textsc{Rothe, C.} (2009): \enquote{{S}emiparametric Estimation of Binary
  Response Models with Endogenous Regressors,} \emph{Journal of Econometrics},
  153, 51--64.

\bibitem[\protect\citeauthoryear{Ruppert and Wand}{Ruppert and
  Wand}{1994}]{ruppertwand1994taos}
\textsc{Ruppert, D. and M.~P. Wand} (1994): \enquote{{M}ultivariate {L}ocally
  {W}eighted {L}east {S}quares {R}egression,} \emph{The Annals of Statistics},
  22, 1346--1370.

\bibitem[\protect\citeauthoryear{Serfling}{Serfling}{1980}]{serfling1980}
\textsc{Serfling, R.~J.} (1980): \emph{{A}pproximation {T}heorems of
  {M}athematical {S}tatistics}, New York: John Wiley \& Sons.

\bibitem[\protect\citeauthoryear{Smith and Huang}{Smith and
  Huang}{1995}]{SmithHuang1995}
\textsc{Smith, V.~K. and J.-C. Huang} (1995): \enquote{Can Markets Value Air
  Quality? A Meta-analysis of Hedonic Property Value Models,} \emph{The Journal
  of Political Economy}, 209--227.

\bibitem[\protect\citeauthoryear{Sperlich}{Sperlich}{2009}]{sperlich2009ej}
\textsc{Sperlich, S.} (2009): \enquote{{A} {N}ote on {N}on-parametric
  {E}stimation with {P}redicted {V}ariables,} \emph{Econometrics Journal}, 12,
  382--395.

\bibitem[\protect\citeauthoryear{Torgovitsky}{Torgovitsky}{2012}]{torgovitsky2012wpa}
\textsc{Torgovitsky, A.} (2012): \enquote{{I}dentification of {N}onseparable
  {M}odels with {G}eneral {I}nstruments,} \emph{Working paper}.

\bibitem[\protect\citeauthoryear{Torgovitsky}{Torgovitsky}{2013}]{torgovitsky2013wp}
---\hspace{-.1pt}---\hspace{-.1pt}--- (2013): \enquote{{M}inimum {D}istance
  from {I}ndependence {E}stimation of {N}onseparable {I}nstrumental {V}ariables
  {M}odels,} \emph{Working paper}.

\bibitem[\protect\citeauthoryear{Train}{Train}{2009}]{train2002}
\textsc{Train, K.} (2009): \emph{{D}iscrete {C}hoice {M}ethods with
  {S}imulation}, Cambridge University Press, 2 ed.

\bibitem[\protect\citeauthoryear{{U.S. Department of Commerce. Bureau of the
  Census}}{{U.S. Department of Commerce. Bureau of the
  Census}}{2008}]{CCDB1983}
\textsc{{U.S. Department of Commerce. Bureau of the Census}} (2008):
  \emph{County and City Data Book [United States], 1983. ICPSR version}, Ann
  Arbor, MI: {Inter-university Consortium for Political and Social Research
  (ICPSR) [distributor]}, \url{http://doi.org/10.3886/ICPSR08256.v1}.

\bibitem[\protect\citeauthoryear{{U.S. Department of Commerce. Bureau of the
  Census}}{{U.S. Department of Commerce. Bureau of the
  Census}}{2012}]{CCDB1977}
---\hspace{-.1pt}---\hspace{-.1pt}--- (2012): \emph{County and City Data Book
  [United States] Consolidated File: County Data, 1947-1977, ICPSR07736-v2},
  Ann Arbor, MI: {Inter-university Consortium for Political and Social Research
  (ICSPR) [distributor]}, \url{http://doi.org/10.3886/ICPSR07736.v2}.

\bibitem[\protect\citeauthoryear{{U.S. Environmental Protection Agency}}{{U.S.
  Environmental Protection Agency}}{2001}]{TSPdata}
\textsc{{U.S. Environmental Protection Agency}} (2001): \enquote{Historical TSP
  data, 1970--1999,}
  \url{http://www.epa.gov/ttn/airs/airsaqs/archived\%20data}, accessed:
  10-8-2013.

\bibitem[\protect\citeauthoryear{{van der Vaart} and Wellner}{{van der Vaart}
  and Wellner}{1996}]{wellner1996}
\textsc{{van der Vaart}, A.~W. and J.~A. Wellner} (1996): \emph{{W}eak
  {C}onvergence and {E}mpirical {P}rocesses: {W}ith {A}pplications to
  {S}tatistics}, Springer.

\bibitem[\protect\citeauthoryear{Wooldridge}{Wooldridge}{1997}]{wooldridge1997el}
\textsc{Wooldridge, J.~M.} (1997): \enquote{{O}n {T}wo {S}tage {L}east
  {S}quares {E}stimation of the {A}verage {T}reatment {E}ffect in a {R}andom
  {C}oefficient {M}odel,} \emph{Economics Letters}, 56, 129--133.

\bibitem[\protect\citeauthoryear{Wooldridge}{Wooldridge}{2003}]{wooldridge2003el}
---\hspace{-.1pt}---\hspace{-.1pt}--- (2003): \enquote{{F}urther {R}esults on
  {I}nstrumental {V}ariables {E}stimation of {A}verage {T}reatment {E}ffects in
  the {C}orrelated {R}andom {C}oefficient {M}odel,} \emph{Economics Letters},
  79, 185--191.

\bibitem[\protect\citeauthoryear{Wooldridge}{Wooldridge}{2008}]{wooldridge2008ic}
---\hspace{-.1pt}---\hspace{-.1pt}--- (2008): \enquote{{I}nstrumental
  {V}ariables {E}stimation of the {A}verage {T}reatment {E}ffect in the
  {C}orrelated {R}andom {C}oefficient {M}odel,} in \emph{Modelling and
  Evaluating Treatment Effects in Econometrics (Advances in Econometrics,
  Volume 21)}, ed. by D.~Millimet, J.~Smith, and E.~Vytlacil, Emerald Group
  Publishing Limited, 93--116.

\end{thebibliography}

\end{document}